\documentclass[smallextended]{svjour3}%
\pdfoutput=1
\usepackage{fullpage}
\usepackage{amsmath}
\usepackage{amsfonts}
\usepackage{amssymb}
\usepackage{graphicx}%
\providecommand{\U}[1]{\protect\rule{.1in}{.1in}}
\journalname{Quantum Information Processing}
\begin{document}

\title{The quantum dynamic capacity formula of a quantum channel\thanks{M.M.W. acknowledges support from the MDEIE (Qu\'{e}bec)
PSR-SIIRI international collaboration grant.}}
\author{Mark M. Wilde
\and Min-Hsiu Hsieh}
\institute{Mark M. Wilde is a postdoctoral fellow with the School of Computer Science, McGill University, Montreal, Quebec, Canada H3A 2A7 \email{mark.wilde@mcgill.ca}
\and Min-Hsiu Hsieh was with the
ERATO-SORST Quantum Computation and Information Project, Japan Science and Technology Agency 5-28-3, Hongo, Bunkyo-ku, Tokyo, Japan during the development of this paper. He is now with the Statistical Laboratory, University of Cambridge, Wilberforce Road, Cambridge, UK CB3 0WB \email{minhsiuh@gmail.com}}
\date{\today}
\maketitle

\begin{abstract}
The dynamic capacity theorem characterizes the reliable communication rates of
a quantum channel when combined with the noiseless resources of classical
communication, quantum communication, and entanglement. In prior work, we
proved the converse part of this theorem by making contact with many previous
results in the quantum Shannon theory literature. In this work, we prove the
theorem with an \textquotedblleft ab initio\textquotedblright\ approach, using
only the most basic tools in the quantum information theorist's toolkit: the
Alicki-Fannes' inequality, the chain rule for quantum mutual information,
elementary properties of quantum entropy, and the quantum data processing
inequality. The result is a simplified proof of the theorem that should be
more accessible to those unfamiliar with the quantum Shannon theory
literature. We also demonstrate that the \textquotedblleft quantum dynamic
capacity formula\textquotedblright\ characterizes the Pareto optimal trade-off
surface for the full dynamic capacity region. Additivity of this formula
simplifies the computation of the trade-off surface,
and we prove that its additivity holds
for the quantum Hadamard channels and the quantum erasure channel. We then
determine exact expressions for and plot the dynamic capacity region of the
quantum dephasing channel, an example from the Hadamard class, and the quantum
erasure channel.
\end{abstract}

\section{Introduction}
Quantum Shannon theory is the study of the transmission capabilities of a
noisy resource when a large number of independent and identically distributed
(IID)\ copies of the resource are available.\footnote{The first few chapters
of Yard's thesis provide an introductory and accessible overview of the
subject \cite{Yard05a}.} An important task in this area of study is to
determine how noiseless resources interact with a noisy quantum channel. That
is, we would like to know the reliable communication rates if a sender can use
noiseless resources in addition to a quantum channel to generate other
noiseless resources. In prior work, we have studied one such setting, where a
sender and receiver generate or consume classical communication, quantum
communication, and entanglement along with the consumption of a noisy quantum
channel~\cite{HW08GFP,HW09book,HW09T3}. The result of these efforts was a
characterization of the \textquotedblleft dynamic capacity
region\textquotedblright\ of a noisy quantum channel.\footnote{``Dynamic'' in this context
and throughout this paper
refers to the fact that a noisy channel is a dynamic resource, as opposed to a ``static''
resource such as a shared bipartite state.}

One of the shortcomings of the characterization of the dynamic capacity region
in Refs.~\cite{HW08GFP,HW09book,HW09T3} is that its computation for a general
quantum channel requires regularized formulas
(which are over an infinite number of channel uses). Though, later, Br\'{a}dler~\textit{et~al}%
.~demonstrated that the quantum Hadamard channels~\cite{KMNR07}\ are a natural
class of channels for which the computation of one octant of the region
simplifies~\cite{BHTW10} because the structure of these channels appears to be
\textquotedblleft just right\textquotedblright\ for this to hold.
The proof considers special quadrants of the dynamic capacity region and
employs several \textit{reductio ad absurdum} arguments to characterize one
octant of the full region. The work of Br\'{a}dler~\textit{et~al}.~showed that
we can claim a complete understanding of the abilities of an
entanglement-assisted quantum Hadamard channel for the transmission of
classical and quantum information.

The aim of the present work is two-fold:\ 1)\ to simplify the proof of the converse part of the
dynamic capacity theorem in Ref.~\cite{HW09T3} and 2) to show that there is
one important formula to consider for any task involving
noiseless classical communication, noiseless quantum communication, noiseless entanglement, and 
many uses of a noisy quantum channel. Our previous proof of the converse
in Ref.~\cite{HW09T3} relies extensively on prior
literature in quantum Shannon theory, perhaps making our ideas inaccessible to an
audience unfamiliar with this increasingly \textquotedblleft tangled
web.\textquotedblright\ Here, we apply an \textquotedblleft ab
initio\textquotedblright\ approach to the proof, using only four tools from
quantum information theory: the Alicki-Fannes'
inequality~\cite{0305-4470-37-5-L01}, the chain rule for quantum mutual
information, elementary properties of quantum entropy, and the quantum data
processing inequality~\cite{SN96}. As such, the proof here should be more
accessible to a broader audience and more streamlined because it is
\textquotedblleft disentangled\textquotedblright\ from the complex web that
the quantum Shannon theory literature has become.

We also propose a new formula that characterizes any task involving classical
communication, quantum communication, and entanglement in dynamic quantum
Shannon theory.\footnote{``Dynamic quantum Shannon theory'' refers to the setting in which
a sender and a receiver have access to many uses of a quantum channel connecting them.}
For this reason, we call this formula the \textquotedblleft
quantum dynamic capacity formula.\textquotedblright\ In particular, additivity
of this formula implies a complete understanding of any task in dynamic
quantum Shannon theory involving the three fundamental noiseless resources. We
find a simplified, direct proof that additivity holds for the Hadamard class
of channels and for a quantum erasure channel. The additivity proof for the quantum erasure
channel is different from that of the Hadamard channel---it exploits the
particular structure of the quantum erasure channel.

We structure this work as follows. The next section reviews the minimal tools
from quantum information theory necessary to understand the rest of the paper.
Section~\ref{sec:info-proc-task} outlines the information processing task considered in this paper,
defines what it means for a rate triple to be achievable, and provides a definition
of the dynamic capacity region of a quantum channel.
Section~\ref{sec:capacity-theorem}\ states the dynamic capacity theorem, and
Section~\ref{sec:achievable}\ contains a brief review of the proof of the
achievability part of the theorem. The main protocol for
proving this part is the \textquotedblleft classically-enhanced father
protocol,\textquotedblright\ whose detailed proof we gave in
Ref.~\cite{HW08GFP}. Section~\ref{sec:converse}\ contains the converse proof,
where we proceed with the minimal tools stated above. We then show in
Section~\ref{sec:dynamic-cap-formula}\ how the quantum dynamic capacity
formula characterizes the optimization task for computing the Pareto optimal
trade-off surface for the dynamic capacity region.
Section~\ref{sec:single-letter-hadamard} proves that the quantum dynamic
capacity formula is additive for the Hadamard
class of channels, and in Section~\ref{sec:dephasing}, we directly
compute and plot the region for a qubit dephasing channel, which is a channel
that falls within the Hadamard class.
Section~\ref{sec:single-letter-erasure} then shows that the quantum
dynamic capacity is additive for the quantum erasure channel, and
we compute and plot the region for this channel also.
Finally, we conclude with a brief discussion.

\section{Definitions and notation}

\label{sec:def-not}We first establish some definitions and notation that we
employ throughout the paper and review a few important properties of quantum
entropy. Let $\Phi^{AB}$ denote the maximally entangled state shared between two
parties:%
\[
\left\vert \Phi\right\rangle ^{AB}\equiv\frac{1}{\sqrt{D}}\sum_{i=1}%
^{D}\left\vert i\right\rangle ^{A}\left\vert i\right\rangle ^{B}.
\]
An ebit corresponds to the special case where $D=2$. Let $\overline{\Phi
}^{M_{A}M_{B}}$ denote the maximally correlated state shared between two
parties:%
\[
\overline{\Phi}^{M_{A}M_{B}}\equiv\frac{1}{D}\sum_{i=1}^{D}\left\vert
i\right\rangle \left\langle i\right\vert ^{M_{A}}\otimes\left\vert
i\right\rangle \left\langle i\right\vert ^{M_{B}}.
\]
A common randomness bit corresponds to the special case where $D=2$.

A completely-positive trace-preserving (CPTP)\ map $\mathcal{N}^{A^{\prime
}\rightarrow B}$\ is the most general map we consider that maps from a quantum
system $A^{\prime}$\ to another quantum system $B$ \cite{NC00}. It acts as follows on any
density operator $\rho$:%
\[
\mathcal{N}^{A^{\prime}\rightarrow B}\left(  \rho\right)  =\sum_{k}A_{k}\rho
A_{k}^{\dag},
\]
where the operators $A_{k}$ satisfy the condition $\sum_{k}A_{k}^{\dag}%
A_{k}=I$. A quantum channel admits an isometric extension $U_{\mathcal{N}%
}^{A^{\prime}\rightarrow BE}$, which is a unitary embedding into a larger
Hilbert space. One recovers the original channel by taking a partial trace
over the \textquotedblleft environment\textquotedblright\ system $E$.

We consider a three-dimensional capacity region throughout this work (as in
Ref.~\cite{HW09T3}), whose points $\left(  C,Q,E\right)  $\ correspond to
rates of classical communication, quantum communication, and entanglement
generation/consumption, respectively. For example, the teleportation protocol
consumes two classical bits and an ebit to generate a noiseless qubit \cite{BBCJPW93}. Thus,
we write it as the following rate triple:%
\[
\left(  -2,1,-1\right)  ,
\]
where we indicate consumption of a resource with a negative sign and
generation of a resource with a positive sign. Also, the super-dense coding
protocol consumes a noiseless qubit channel and an ebit to generate two
classical bits \cite{BW92}. It corresponds to the rate triple:%
\[
\left(  2,-1,-1\right)  .
\]
Another protocol that we exploit is entanglement distribution. It uses a
noiseless qubit channel to establish a noiseless ebit and corresponds to%
\[
\left(  0,-1,1\right)  .
\]

The entropy $H\left(  A\right)  _{\rho}$ of a density operator $\rho^{A}$ on
some quantum system $A$ is as follows \cite{NC00}:%
\[
H\left(  A\right)  _{\rho}\equiv-\text{Tr}\left\{  \rho^{A}\log\rho
^{A}\right\}  ,
\]
where the logarithm is base two. The entropy can never exceed the logarithm of
the dimension of $A$. The quantum mutual information $I\left(  A;B\right)
_{\sigma}$\ of a bipartite state $\sigma^{AB}$ is as follows:%
\[
I\left(  A;B\right)  _{\sigma}\equiv H\left(  A\right)  _{\sigma}+H\left(
B\right)  _{\sigma}-H\left(  AB\right)  _{\sigma}.
\]
Observe that the quantum mutual information $I\left(  M_{A};M_{B}\right)
$\ of the state $\overline{\Phi}^{M_{A}M_{B}}$ is equal to $\log D$ bits, and
the quantum mutual information $I\left(  A;B\right)  $\ of the state
$\Phi^{AB}$ is equal to $2\log D$ qubits. If one system is classical, then the
quantum mutual information can never be greater than the logarithm of the
dimension of the classical system. If both systems are quantum, then the
quantum mutual information can never be greater than twice the minimum of the
logarithms of the dimensions of the two quantum systems. The quantum mutual
information vanishes if the bipartite state $\sigma^{AB}$\ is a product state.
The conditional quantum mutual information for three quantum systems $A$, $B$,
and $C$ is as follows:%
\[
I\left(  A;B|C\right)  \equiv H\left(  AC\right)  +H\left(  BC\right)
-H\left(  C\right)  -H\left(  ABC\right)  ,
\]
and is always non-negative due to strong subadditivity~\cite{LR73}. The
coherent information $I\left(  A\rangle B\right)  _{\sigma}$\ of a state
$\sigma^{AB}$ is as follows \cite{SN96}:%
\[
I\left(  A\rangle B\right)  _{\sigma}\equiv H\left(  B\right)  _{\sigma
}-H\left(  AB\right)  _{\sigma}.
\]
Observe that the coherent information $I\left(  A\rangle B\right)  $\ of the
state $\Phi^{AB}$ is equal to $\log D$ qubits. The chain rule for quantum
mutual information gives the following relation for any three quantum
systems $A$, $B$, and $C$:%
\begin{align}
I\left(  AB;C\right)   &  =I\left(  B;C\right)  +I\left(  A;C|B\right)
\label{eq:mut-chain-rule}\\
&  =I\left(  A;C\right)  +I\left(  B;C|A\right)  .\nonumber
\end{align}

A classical-quantum state $\sigma^{XABE}$\ of the following form plays an
important role throughout this paper:%
\[
\sigma^{XABE}\equiv\sum_{x}p_{X}\left(  x\right)  \left\vert x\right\rangle
\left\langle x\right\vert ^{X}\otimes U_{\mathcal{N}}^{A^{\prime}\rightarrow
BE}(\phi_{x}^{AA^{\prime}}),
\]
where the states $\phi_{x}^{AA^{\prime}}$ are pure bipartite states and
$U_{\mathcal{N}}^{A^{\prime}\rightarrow BE}$ is the isometric extension of
some noisy channel $\mathcal{N}^{A^{\prime}\rightarrow B}$. Applying the above
chain rule gives the following relation:%
\begin{equation}
I\left(  AX;B\right)  _{\sigma}=I\left(  X;B\right)  _{\sigma}+I\left(
A;B|X\right)  _{\sigma} . \label{eq:entropy-1}%
\end{equation}
Additionally, one can readily check that the following relation holds for a
state of the above form:%
\begin{equation}
I\left(  A\rangle BX\right)  _{\sigma}=\frac{1}{2}I\left(  A;B|X\right)
_{\sigma}-\frac{1}{2}I\left(  A;E|X\right)  _{\sigma}. \label{eq:entropy-2}%
\end{equation}

The Alicki-Fannes' inequality is a statement of the continuity of coherent
information~\cite{0305-4470-37-5-L01}, and a simple variant of it gives
continuity of quantum mutual information. First, suppose that two bipartite
states $\rho^{AB}$ and $\sigma^{AB}$ are $\epsilon$-close in trace norm:%
\[
\left\Vert \rho^{AB}-\sigma^{AB}\right\Vert _{1}\leq\epsilon\text{.}%
\]
Then the Alicki-Fannes' inequality states that their respective coherent
informations are close:%
\[
\left\vert I(A\rangle B)_{\rho}-I(A\rangle B)_{\sigma}\right\vert
\leq4\epsilon\log\left\vert A\right\vert +2H_{2}(\epsilon),
\]
where $\left\vert A\right\vert $ is the dimension of the system $A$ and
$H_{2}$ is the binary entropy function. A simple \textquotedblleft
tweak\textquotedblright\ of the above inequality shows that the quantum mutual
informations are close~\cite{HW08GFP}:%
\[
\left\vert I(A;B)_{\rho}-I(A;B)_{\sigma}\right\vert \leq5\epsilon
\log\left\vert A\right\vert +3H_{2}(\epsilon).
\]

The quantum data processing inequality states that quantum correlations can
never increase under the application of a noisy map~\cite{SN96,NC00}. There are two
important manifestations of it. Suppose that Alice possesses two quantum
systems $A$ and $A^{\prime}$\ in her lab. She sends $A^{\prime}$ through a
noisy quantum channel to Bob and he receives it in some system $B$. Then the
following inequality applies to the coherent information:%
\[
I\left(  A\rangle A^{\prime}\right)  \geq I\left(  A\rangle B\right)  ,
\]
demonstrating that quantum correlations, as measured by the coherent
information, can only decrease under noisy processing. The following
inequality also applies to the quantum mutual information:%
\[
I\left(  A;A^{\prime}\right)  \geq I\left(  A;B\right)  ,
\]
demonstrating a similar notion for the quantum mutual information. Now suppose
that Alice sends $A$ to Christabelle and she receives it in some system $C$.
Then the following inequality applies as well:%
\[
I\left(  A;A^{\prime}\right)  \geq I\left(  C;A^{\prime}\right)  \geq I\left(
C;B\right)  ,
\]
but a similar inequality does not hold for the coherent information, due to
its asymmetry. We make extensive use of quantum data processing in our proofs.

\section{The information processing task}

\label{sec:info-proc-task}

\begin{figure}
[ptb]
\begin{center}
\includegraphics[
natheight=5.826200in,
natwidth=8.260700in,
height=4.2748in,
width=6.0502in
]%
{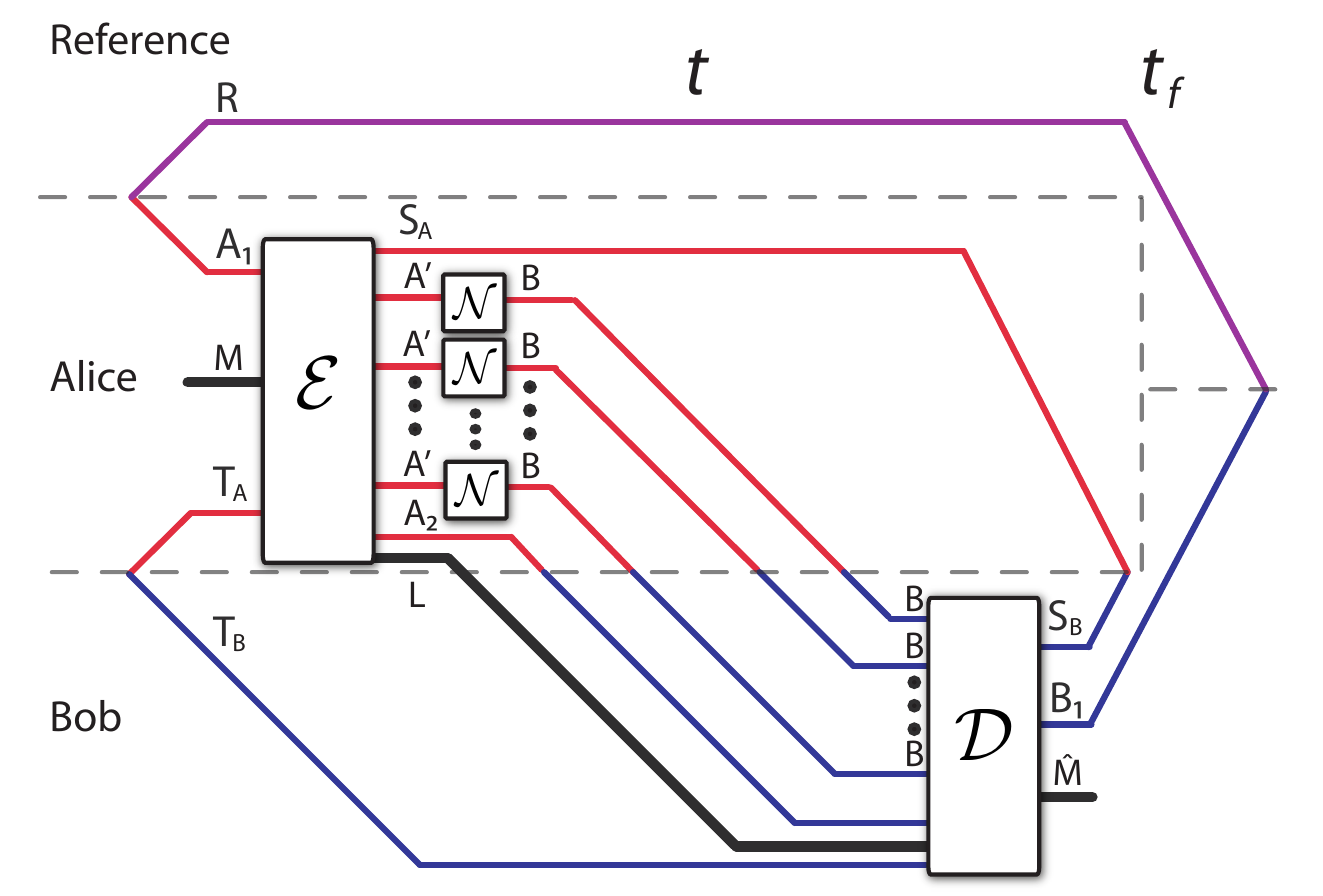}%
\caption{(Color online) The most general protocol for generating classical
communication, quantum communication, and entanglement with the
help of the same respective resources and many uses of a noisy quantum
channel. Alice begins with her classical register $M$, her quantum register
$A_{1}$, and her half of the shared entanglement in register $T_{A}$. She
encodes according to some CPTP map$\ \mathcal{E}$ that outputs a quantum
register $S_{A}$, many registers $A^{\prime n}$, a quantum register $A_{2}$,
and a classical register~$L$. She inputs $A^{\prime n}$ to many uses of the
noisy channel $\mathcal{N}$ and transmits $A_{2}$ over a noiseless quantum
channel and $L$ over a noiseless classical channel. Bob receives the channel
outputs $B^{n}$, the quantum register $A_{2}$, and the classical register $L$
and performs a decoding $\mathcal{D}$ that recovers the quantum information
and classical message. The decoding also generates entanglement with system
$S_{A}$. Many protocols are a special case of the above one. For example, the
protocol is entanglement-assisted communication of classical and quantum
information~\cite{HW08GFP}\ if the registers $L$, $S_{A}$, $S_{B}$, and
$A_{2}$ are null.}%
\label{fig:catalytic-protocol}%
\end{center}
\end{figure}

We are interested in the most general protocol that
generates classical communication, quantum communication, and entanglement
by consuming many uses of a noisy quantum channel $\mathcal{N}^{A^{\prime
}\rightarrow B}$ and the same respective resources 
(see Figure~\ref{fig:catalytic-protocol}). We say that such a
protocol is ``catalytic'' because we are allowing it to consume the same resources
that it generates, though we are keeping track of the net rates of consumption or generation.

The protocol begins with Alice possessing two
classical registers (each labeled by $M$ and of dimension $2^{n\bar{C}}$), a
quantum register $A_{1}$ of dimension $2^{n\bar{Q}}$ entangled with a
reference system $R$, and another quantum register $T_{A}$ of dimension
$2^{n\tilde{E}}$ that contains her half of the shared entanglement with Bob:%
\[
\omega^{MMRA_{1}T_{A}T_{B}}\equiv\overline{\Phi}^{MM}\otimes\Phi^{RA_{1}%
}\otimes\Phi^{T_{A}T_{B}}.
\]
She passes one of the classical registers and the registers $A_{1}$ and
$T_{A}$ into a CPTP\ encoding map $\mathcal{E}^{MA_{1}T_{A}\rightarrow
A^{\prime n}S_{A}LA_{2}}$ that outputs a quantum register $S_{A}$ of dimension
$2^{n\bar{E}}$ and a quantum register $A_{2}$ of dimension $2^{n\tilde{Q}}$, a
classical register $L$ of dimension $2^{n\tilde{C}}$, and many quantum systems
$A^{\prime n}$ for input to the channel. The register $S_{A}$\ is for creating
entanglement with Bob. The state after the encoding map $\mathcal{E}$\ is as
follows:%
\[
\omega^{MA^{\prime n}S_{A}LA_{2}RT_{B}}\equiv\mathcal{E}^{MA_{1}%
T_{A}\rightarrow A^{\prime n}S_{A}LA_{2}}(\omega^{MMRA_{1}T_{A}T_{B}}).
\]
She sends the systems $A^{\prime n}$ through many uses $\mathcal{N}^{A^{\prime
n}\rightarrow B^{n}}$ of the noisy channel $\mathcal{N}^{A^{\prime}\rightarrow
B}$, transmits $L$ over a noiseless classical channel, and transmits $A_{2}$
over a noiseless quantum channel, producing the following state:%
\[
\omega^{MB^{n}S_{A}LA_{2}RT_{B}}\equiv\mathcal{N}^{A^{\prime n}\rightarrow
B^{n}}(\omega^{MA^{\prime n}S_{A}LA_{2}RT_{B}}).
\]
The above state is a state of the form in (\ref{eq:converse-state}) with
$A\equiv RT_{B}A_{2}S_{A}$ and $X\equiv ML$. Bob then applies a map
$\mathcal{D}^{B^{n}A_{2}T_{B}L\rightarrow B_{1}S_{B}\hat{M}}$ that outputs a
quantum system $B_{1}$, a quantum system $S_{B}$, and a classical register
$\hat{M}$. Let $\omega^{\prime}$ denote the final state. The following
condition holds for a good protocol:%
\begin{equation}
\left\Vert \overline{\Phi}^{M\hat{M}}\otimes\Phi^{RB_{1}%
}\otimes\Phi^{S_{A}S_{B}}-\left(  \omega^{\prime}\right)  ^{MB_{1}S_{B}\hat
{M}S_{A}R}\right\Vert _{1}\leq\epsilon,\label{eq:+++_good-code}%
\end{equation}
implying that Alice and Bob establish maximal classical correlations in $M$
and $\hat{M}$ and maximal entanglement between $S_{A}$ and $S_{B}$. The above
condition also implies that the coding scheme preserves the entanglement with
the reference system $R$. The net rate triple for the protocol is as follows:
$\left(  \bar{C}-\tilde{C},\bar{Q}-\tilde{Q},\bar{E}-\tilde{E}\right)  $. The
protocol generates a resource if its corresponding rate is positive, and it
consumes a resource if its corresponding rate is negative. Such a protocol defines an 
$(n, C, Q, E, \epsilon)$ code with
\begin{align}C & = \bar{C}-\tilde{C}, \\
 Q & = \bar{Q}-\tilde{Q},\\
 E & =\bar{E}-\tilde{E} .\end{align}

\begin{definition}[Achievability] A rate triple $(C,Q,E)$ is achievable
if there exists an $(n, C, Q, E, \epsilon)$ code with error (as defined in (\ref{eq:+++_good-code})) smaller than
 $\epsilon$ for all $\epsilon > 0$ and sufficiently large $n$.
 \end{definition}

\begin{definition}[Dynamic Capacity Region]
The dynamic capacity region $\mathcal{C}_{\text{CQE}}\left(
\mathcal{N}\right)  $ of a noisy quantum channel $\mathcal{N}^{A^{\prime}\rightarrow
B}$ is a three-dimensional region in the $(C,Q,E)$ space defined by the
closure of the set of all achievable rate triples $(C,Q,E)$.
\end{definition}


\section{The dynamic capacity theorem}

\label{sec:capacity-theorem}The dynamic capacity theorem gives bounds on the
reliable communication rates of a noisy quantum channel when combined with the
noiseless resources of classical communication, quantum communication, and
shared entanglement~\cite{HW09T3}. The theorem applies regardless of whether a
protocol consumes the noiseless resources or generates them.

\begin{theorem}
[Dynamic Capacity]\label{thm:main-theorem}The dynamic capacity region
$\mathcal{C}_{\text{\emph{CQE}}}(\mathcal{N})$ of a quantum channel
$\mathcal{N}$ is equal to the following expression:%
\begin{equation}
\mathcal{C}_{\text{\emph{CQE}}}(\mathcal{N})=\overline{\bigcup_{k=1}^{\infty
}\frac{1}{k}\mathcal{C}_{\text{\emph{CQE}}}^{(1)}(\mathcal{N}^{\otimes k})},
\label{eq:multi-letter}%
\end{equation}
where the overbar indicates the closure of a set. The \textquotedblleft
one-shot\textquotedblright\ region $\mathcal{C}_{\text{\emph{CQE}}}%
^{(1)}(\mathcal{N})$ is the union of the \textquotedblleft one-shot,
one-state\textquotedblright\ regions $\mathcal{C}_{\text{\emph{CQE}},\sigma
}^{(1)}(\mathcal{N})$:%
\[
\mathcal{C}_{\text{\emph{CQE}}}^{(1)}(\mathcal{N})\equiv\bigcup_{\sigma
}\mathcal{C}_{\text{\emph{CQE}},\sigma}^{(1)}(\mathcal{N}).
\]
The \textquotedblleft one-shot, one-state\textquotedblright\ region
$\mathcal{C}_{\text{\emph{CQE}},\sigma}^{(1)}(\mathcal{N})$ is the set of all
rates $C$, $Q$, and $E$, such that%
\begin{align}
C+2Q  &  \leq I(AX;B)_{\sigma},\label{eq:CQ-bound}\\
Q+E  &  \leq I(A\rangle BX)_{\sigma},\label{eq:QE-bound}\\
C+Q+E  &  \leq I(X;B)_{\sigma}+I(A\rangle BX)_{\sigma}. \label{eq:CQE-bound}%
\end{align}
The above entropic quantities are with respect to a classical-quantum state
$\sigma^{XAB}$ where%
\begin{equation}
\sigma^{XAB}\equiv\sum_{x}p(x)\left\vert x\right\rangle \left\langle
x\right\vert ^{X}\otimes\mathcal{N}^{A^{\prime}\rightarrow B}(\phi
_{x}^{AA^{\prime}}), \label{eq:main-theorem-state}%
\end{equation}
and the states $\phi_{x}^{AA^{\prime}}$ are pure. It is implicit that one
should consider states on $A^{\prime k}$ instead of $A^{\prime}$ when taking
the regularization in (\ref{eq:multi-letter}).
\end{theorem}

The above theorem is a \textquotedblleft multi-letter\textquotedblright%
\ capacity theorem because of the regularization in (\ref{eq:multi-letter}).
Though, we show in Sections~\ref{sec:single-letter-hadamard}\ and
\ref{sec:single-letter-erasure}\ that the regularization is not necessary for
the Hadamard class of channels or the quantum erasure channels, respectively.
We prove the above theorem in two parts:

\begin{enumerate}
\item The direct coding theorem below shows that combining the
\textquotedblleft classically-enhanced father protocol\textquotedblright\ with
teleportation, super-dense coding, and entanglement distribution achieves the
above region.

\item The converse theorem demonstrates that any coding scheme cannot do
better than the regularization in (\ref{eq:multi-letter}), in the sense that a
scheme with vanishing error should have its rates below the above amounts. We
prove the converse theorem directly in \textquotedblleft one fell
swoop,\textquotedblright\ by employing a catalytic, information-theoretic
approach. The converse proof is different from our earlier one~\cite{HW09T3}%
\ because we employ straightforward information-theoretic arguments instead of
making contact with prior quantum Shannon theoretic literature.
\end{enumerate}

\section{Dynamic achievable rate region}

\label{sec:achievable}The unit resource achievable region is what Alice and
Bob can achieve with the protocols entanglement distribution, teleportation,
and super-dense coding~\cite{HW09T3}. It is the cone of the rate triples
corresponding to these protocols:%
\[
\left\{  \alpha\left(  0,-1,1\right)  +\beta\left(  2,-1,-1\right)
+\gamma\left(  -2,1,-1\right)  :\alpha,\beta,\gamma\geq0\right\}  .
\]
We can also write any rate triple $\left(  C,Q,E\right)  $ in the unit
resource capacity region with a matrix equation:%
\begin{equation}%
\begin{bmatrix}
C\\
Q\\
E
\end{bmatrix}
=%
\begin{bmatrix}
0 & 2 & -2\\
-1 & -1 & 1\\
1 & -1 & -1
\end{bmatrix}%
\begin{bmatrix}
\alpha\\
\beta\\
\gamma
\end{bmatrix}
.\label{eq:unit-resource-achievable-region}%
\end{equation}
The inverse of the above matrix is as follows:%
\[%
\begin{bmatrix}
-\frac{1}{2} & -1 & 0\\
0 & -\frac{1}{2} & -\frac{1}{2}\\
-\frac{1}{2} & -\frac{1}{2} & -\frac{1}{2}%
\end{bmatrix}
,
\]
and gives the following set of inequalities for the unit resource achievable
region:%
\begin{align*}
C+2Q &  \leq0,\\
Q+E &  \leq0,\\
C+Q+E &  \leq0,
\end{align*}
by inverting the matrix equation in (\ref{eq:unit-resource-achievable-region})
and applying the constraints $\alpha,\beta,\gamma\geq0$.

Now, let us include the classically-enhanced father protocol~\cite{HW08GFP}.
Ref.~\cite{HW08GFP} proved that we can achieve the following rate triple by
channel coding over a noisy quantum channel $\mathcal{N}^{A^{\prime
}\rightarrow B}$:%
\[
\left(  I\left(  X;B\right)  _{\sigma},\frac{1}{2}I\left(  A;B|X\right)
_{\sigma},-\frac{1}{2}I\left(  A;E|X\right)  _{\sigma}\right)  ,
\]
for any state $\sigma^{XABE}$\ of the form:%
\begin{equation}
\sigma^{XABE}\equiv\sum_{x}p_{X}\left(  x\right)  \left\vert x\right\rangle
\left\langle x\right\vert ^{X}\otimes U_{\mathcal{N}}^{A^{\prime}\rightarrow
BE}(\phi_{x}^{AA^{\prime}}),\label{eq:maximization-state}%
\end{equation}
where $U_{\mathcal{N}}^{A^{\prime}\rightarrow BE}$ is an isometric extension
of the quantum channel $\mathcal{N}^{A^{\prime}\rightarrow B}$. Specifically,
we showed in Ref.~\cite{HW08GFP}\ that one can achieve the above rates with
vanishing error in the limit of large blocklength. Thus the achievable rate
region is the following translation of the unit resource achievable region in
(\ref{eq:unit-resource-achievable-region}):%
\[%
\begin{bmatrix}
C\\
Q\\
E
\end{bmatrix}
=%
\begin{bmatrix}
0 & 2 & -2\\
-1 & -1 & 1\\
1 & -1 & -1
\end{bmatrix}%
\begin{bmatrix}
\alpha\\
\beta\\
\gamma
\end{bmatrix}
+%
\begin{bmatrix}
I\left(  X;B\right)  _{\sigma}\\
\frac{1}{2}I\left(  A;B|X\right)  _{\sigma}\\
-\frac{1}{2}I\left(  A;E|X\right)  _{\sigma}%
\end{bmatrix}
.
\]
We can now determine bounds on an achievable rate region that employs the
above coding strategy. We apply the inverse of the matrix in
(\ref{eq:unit-resource-achievable-region}) to the LHS and RHS. Then using
(\ref{eq:entropy-1}), (\ref{eq:entropy-2}), and the constraints $\alpha
,\beta,\gamma\geq0$, we obtain the inequalities in (\ref{eq:CQ-bound}%
-\ref{eq:CQE-bound}), corresponding exactly to the one-shot, one-state region
in Theorem~\ref{thm:main-theorem}. Taking the union over all possible states
$\sigma$ in (\ref{eq:maximization-state}) and taking the regularization gives
the full dynamic achievable rate region.

\section{Catalytic and information theoretic converse proof}

\label{sec:converse}This section begins one of the main contributions of this
work. We provide a catalytic, information theoretic converse proof of the
dynamic capacity region, showing that (\ref{eq:multi-letter}) gives a
multi-letter characterization of it. The catalytic approach means that we are
considering the most general protocol that \textit{consumes and generates}
classical communication, quantum communication, and entanglement in addition
to the uses of the noisy quantum channel. This approach has the advantage that
we can prove the converse theorem in \textquotedblleft one fell
swoop\textquotedblright\ rather than considering one octant of the $\left(
C,Q,E\right)  $ space at a time as we did in Ref.~\cite{HW09T3}. Additionally,
we do not need to make contact with prior work in quantum Shannon theory. We
employ the Alicki-Fannes' inequality, the chain rule for quantum mutual
information, elementary properties of quantum entropy, and the quantum data
processing inequality to prove the converse.

There are some subtleties in our proof for the converse theorem. We prove that
the bounds in (\ref{eq:CQ-bound}-\ref{eq:CQE-bound}) hold for common
randomness generation instead of classical communication because a capacity
for generating common randomness can only be better than that for generating
classical communication (classical communication can generate common
randomness). We also consider a protocol that preserves entanglement with a
reference system instead of one that generates quantum communication.
Barnum~\textit{et~al}.~showed that this task is equivalent to the transmission
of quantum information~\cite{BKN98}.

We prove that the converse theorem holds for a state of the following form%
\begin{equation}
\sigma^{XAB}\equiv\sum_{x}p(x)\left\vert x\right\rangle \left\langle
x\right\vert ^{X}\otimes\mathcal{N}^{A^{\prime}\rightarrow B}(\rho
_{x}^{AA^{\prime}}),\label{eq:converse-state}%
\end{equation}
where the states $\rho_{x}^{AA^{\prime}}$ are mixed, rather than proving it
for a state of the form in (\ref{eq:maximization-state}). Then we show in
Section~\ref{sec:pure-suff}\ that it is not necessary to consider an ensemble
of mixed states---i.e., we can do just as well with an ensemble of pure
states, giving the statement of Theorem~\ref{thm:main-theorem}.%

We
begin by proving the first bound in (\ref{eq:CQ-bound}). All system labels
are as given in Section~\ref{sec:info-proc-task}, and we consider the most general
protocol as outlined in that section. Consider the following chain
of inequalities:%
\begin{align*}
n\left(  \bar{C}+2\bar{Q}\right)   &  =I(M;\hat{M})_{\overline{\Phi}}+I\left(
R;B_{1}\right)  _{\Phi}\\
&  \leq I(M;\hat{M})_{\omega^{\prime}}+I\left(  R;B_{1}\right)  _{\omega
^{\prime}}+n\delta^{\prime}\\
&  \leq I\left(  M;B^{n}A_{2}LT_{B}\right)  _{\omega}+I\left(  R;B^{n}%
A_{2}LT_{B}\right)  _{\omega}\\
&  \leq I\left(  M;B^{n}A_{2}LT_{B}\right)  _{\omega}+I\left(  R;B^{n}%
A_{2}LT_{B}M\right)  _{\omega}\\
&  =I\left(  M;B^{n}A_{2}LT_{B}\right)  _{\omega}+I\left(  R;B^{n}A_{2}%
LT_{B}|M\right)  _{\omega}+I\left(  R;M\right)  _{\omega}.
\end{align*}
The first equality holds by evaluating the quantum mutual informations on the
respective states $\overline{\Phi}^{M\hat{M}}$ and $\Phi^{RB_{1}}$. The first
inequality follows from the condition in (\ref{eq:+++_good-code}) and an
application of the Alicki-Fannes' inequality where $\delta^{\prime}$ vanishes
as $\epsilon\rightarrow0$. We suppress this term in the rest of the
inequalities for convenience. The second inequality follows from quantum data
processing, and the third follows from another application of quantum data
processing. The second equality follows by applying the mutual information
chain rule in (\ref{eq:mut-chain-rule}). We continue below:%
\begin{align*}
&  =I\left(  M;B^{n}A_{2}LT_{B}\right)  _{\omega}+I\left(  R;B^{n}A_{2}%
LT_{B}|M\right)  _{\omega}\\
&  =I\left(  M;B^{n}A_{2}LT_{B}\right)  _{\omega}+I\left(  RA_{2}LT_{B}%
;B^{n}|M\right)  _{\omega}+I\left(  R;A_{2}LT_{B}|M\right)  _{\omega}-I\left(
B^{n};A_{2}LT_{B}|M\right)  _{\omega}\\
&  =I\left(  M;B^{n}\right)  _{\omega}+I\left(  M;A_{2}LT_{B}|B^{n}\right)
_{\omega}+I\left(  RA_{2}LT_{B};B^{n}|M\right)  _{\omega}+I\left(
R;A_{2}LT_{B}|M\right)  _{\omega}-I\left(  B^{n};A_{2}LT_{B}|M\right)
_{\omega}\\
&  =I\left(  RA_{2}LMT_{B};B^{n}\right)  _{\omega}+I\left(  M;A_{2}%
LT_{B}|B^{n}\right)  _{\omega}+I\left(  R;A_{2}LT_{B}|M\right)  _{\omega
}-I\left(  B^{n};A_{2}LT_{B}|M\right)  _{\omega}\\
&  \leq I\left(  RA_{2}T_{B}S_{A}LM;B^{n}\right)  _{\omega}+I\left(
M;A_{2}LT_{B}|B^{n}\right)  _{\omega}+I\left(  R;A_{2}LT_{B}|M\right)
_{\omega}-I\left(  B^{n};A_{2}LT_{B}|M\right)  _{\omega}\\
&  =I\left(  AX;B^{n}\right)  _{\omega}+I\left(  M;A_{2}LT_{B}|B^{n}\right)
_{\omega}+I\left(  R;A_{2}LT_{B}|M\right)  _{\omega}-I\left(  B^{n}%
;A_{2}LT_{B}|M\right)  _{\omega}.
\end{align*}
The first equality follows because $I\left(  R;M\right)  _{\omega}=0$ for this
protocol. The second equality follows from applying the chain rule for quantum
mutual information to the term $I\left(  R;B^{n}A_{2}LT_{B}|M\right)
_{\omega}$, and the third is another application of the chain rule to the term
$I\left(  M;B^{n}A_{2}LT_{B}\right)  _{\omega}$. The fourth equality follows
by combining $I\left(  M;B^{n}\right)  _{\omega}$ and $I\left(  RA_{2}%
LT_{B};B^{n}|M\right)  _{\omega}$ with the chain rule. The inequality follows
from an application of quantum data processing. The final equality follows
from the definitions $A\equiv RT_{B}A_{2}S_{A}$ and $X\equiv ML$. We now focus
on the term $I\left(  M;A_{2}LT_{B}|B^{n}\right)  _{\omega}+I\left(
R;A_{2}LT_{B}|M\right)  _{\omega}-I\left(  B^{n};A_{2}LT_{B}|M\right)
_{\omega}$ and show that it is less than $n\left(  \tilde{C}+2\tilde
{Q}\right)  $:%
\begin{align*}
&  I\left(  M;A_{2}LT_{B}|B^{n}\right)  _{\omega}+I\left(  R;A_{2}%
LT_{B}|M\right)  _{\omega}-I\left(  B^{n};A_{2}LT_{B}|M\right)  _{\omega}\\
&  =I\left(  M;A_{2}LT_{B}B^{n}\right)  _{\omega}+I\left(  R;A_{2}%
LT_{B}M\right)  _{\omega}-I\left(  B^{n};A_{2}LT_{B}M\right)  _{\omega
}-I\left(  R;M\right)  _{\omega}\\
&  =I\left(  M;A_{2}LT_{B}B^{n}\right)  _{\omega}+I\left(  R;A_{2}%
LT_{B}M\right)  _{\omega}-I\left(  B^{n};A_{2}LT_{B}M\right)  _{\omega}\\
&  =H\left(  A_{2}LT_{B}B^{n}\right)  _{\omega}+H\left(  R\right)  _{\omega
}-H\left(  RA_{2}LT_{B}|M\right)  _{\omega}-H\left(  B^{n}\right)  _{\omega}\\
&  =H\left(  A_{2}LT_{B}B^{n}\right)  _{\omega}-H\left(  A_{2}LT_{B}%
|MR\right)  _{\omega}-H\left(  B^{n}\right)  _{\omega}\\
&  =H\left(  A_{2}LT_{B}|B^{n}\right)  _{\omega}-H\left(  A_{2}LT_{B}%
|MR\right)  _{\omega}.
\end{align*}
The first equality follows by applying the chain rule for quantum mutual
information. The second equality follows because $I\left(  R;M\right)
_{\omega}=0$ for this protocol. The third equality follows by expanding the
quantum mutual informations. The next two inequalities follow from
straightforward entropic manipulations and that $H\left(  R\right)  _{\omega
}=H\left(  R|M\right)  _{\omega}$ for this protocol. We continue below:%
\begin{align*}
&  =H\left(  A_{2}L|B^{n}\right)  _{\omega}+H\left(  T_{B}|B^{n}A_{2}L\right)
_{\omega}-H\left(  T_{B}|MR\right)  _{\omega}-H\left(  A_{2}L|T_{B}MR\right)
_{\omega}\\
&  =H\left(  A_{2}L|B^{n}\right)  _{\omega}+H\left(  T_{B}|B^{n}A_{2}L\right)
_{\omega}-H\left(  T_{B}\right)  _{\omega}-H\left(  A_{2}L|T_{B}MR\right)
_{\omega}\\
&  =H\left(  A_{2}L|B^{n}\right)  _{\omega}-I\left(  T_{B};B^{n}A_{2}L\right)
_{\omega}-H\left(  A_{2}L|T_{B}MR\right)  \\
&  \leq H\left(  A_{2}L\right)  _{\omega}-H\left(  A_{2}L|T_{B}MR\right)
_{\omega}\\
&  =I\left(  A_{2}L;T_{B}MR\right)  _{\omega}\\
&  =I\left(  L;T_{B}MR\right)  _{\omega}+I\left(  A_{2};T_{B}MR|L\right)
_{\omega}\\
&  \leq n\left(  \tilde{C}+2\tilde{Q}\right)  .
\end{align*}
The first two equalities follow from the chain rule for entropy and the second
exploits that $H\left(  T_{B}|MR\right)  =H\left(  T_{B}\right)  $ for this
protocol. The third equality follows from the definition of quantum mutual
information. The inequality follows from subadditivity of entropy and that
$I\left(  T_{B};B^{n}A_{2}L\right)  _{\omega}\geq0$. The fourth equality
follows from the definition of quantum mutual information and the next
equality follows from the chain rule. The final inequality follows because the
quantum mutual information $I\left(  L;T_{B}MR\right)  _{\omega}$ can never be
larger than the logarithm of the dimension of the classical register $L$ and
because the quantum mutual information $I\left(  A_{2};T_{B}MR|L\right)
_{\omega}$ can never be larger than twice the logarithm of the dimension of
the quantum register $A_{2}$. Thus the following inequality applies%
\[
n\left(  \bar{C}+2\bar{Q}\right)  \leq I\left(  AX;B^{n}\right)  _{\omega
}+n\left(  \tilde{C}+2\tilde{Q}\right)  +n\delta^{\prime},
\]
demonstrating that (\ref{eq:CQ-bound}) holds for the net rates.

We now prove the second bound in (\ref{eq:QE-bound}). Consider the following
chain of inequalities:%
\begin{align*}
n\left(  \bar{Q}+\bar{E}\right)   &  =I\left(  R\rangle B_{1}\right)  _{\Phi
}+I\left(  S_{A}\rangle S_{B}\right)  _{\Phi}\\
&  =I\left(  RS_{A}\rangle B_{1}S_{B}\right)  _{\Phi\otimes\Phi}\\
&  \leq I\left(  RS_{A}\rangle B_{1}S_{B}\right)  _{\omega^{\prime}}%
+n\delta^{\prime}\\
&  \leq I\left(  RS_{A}\rangle B_{1}S_{B}M\right)  _{\omega^{\prime}}\\
&  \leq I\left(  RS_{A}\rangle B^{n}A_{2}T_{B}LM\right)  _{\omega}\\
&  =H\left(  B^{n}A_{2}T_{B}|LM\right)  _{\omega}-H\left(  RS_{A}B^{n}%
A_{2}T_{B}|LM\right)  _{\omega}\\
&  \leq H\left(  B^{n}|LM\right)  _{\omega}+H\left(  A_{2}|LM\right)
_{\omega}+H\left(  T_{B}|LM\right)  _{\omega}-H\left(  RS_{A}B^{n}A_{2}%
T_{B}|LM\right)  _{\omega}\\
&  \leq I\left(  RS_{A}A_{2}T_{B}\rangle B^{n}LM\right)  _{\omega}+n\left(
\tilde{Q}+\tilde{E}\right)  \\
&  =I\left(  A\rangle B^{n}X\right)  _{\omega}+n\left(  \tilde{Q}+\tilde
{E}\right)  .
\end{align*}
The first equality follows by evaluating the coherent informations of the
respective states $\Phi^{RB_{1}}$ and $\Phi^{S_{A}S_{B}}$. The second equality
follows because $\Phi^{RB_{1}}\otimes\Phi^{T_{A}T_{B}}$ is a product state.
The first inequality follows from the condition in (\ref{eq:+++_good-code})
and an application of the Alicki-Fannes' inequality with $\delta^{\prime}$
vanishing when $\epsilon\rightarrow0$. We suppress the term $n\delta^{\prime}$
in the following lines. The next two inequalities follow from quantum data
processing. The third equality follows from the definition of coherent
information. The fourth inequality follows from subadditivity of entropy. The
fifth inequality follows from the definition of coherent information and the
fact that the entropy can never be larger than the logarithm of the dimension
of the corresponding system. The final equality follows from the definitions
$A\equiv RT_{B}A_{2}S_{A}$ and $X\equiv ML$. Thus the following inequality
applies%
\[
n\left(  \bar{Q}+\bar{E}\right)  \leq I\left(  A\rangle B^{n}X\right)
+n\left(  \tilde{Q}+\tilde{E}\right)  ,
\]
demonstrating that (\ref{eq:QE-bound}) holds for the net rates.

We prove the last bound in (\ref{eq:CQE-bound}). Consider the following chain
of inequalities:%
\begin{align*}
n\left(  \bar{C}+\bar{Q}+\bar{E}\right)   &  =I(M;\hat{M})_{\overline{\Phi}%
}+I\left(  RS_{A}\rangle B_{1}S_{B}\right)  _{\Phi\otimes\Phi}\\
&  \leq I(M;\hat{M})_{\omega^{\prime}}+I\left(  RS_{A}\rangle B_{1}%
S_{B}\right)  _{\omega^{\prime}}+n\delta^{\prime}\\
&  \leq I(M;B^{n}A_{2}T_{B}L)_{\omega}+I\left(  RS_{A}\rangle B^{n}A_{2}%
T_{B}LM\right)  _{\omega}\\
&  =I\left(  ML;B^{n}A_{2}T_{B}\right)  _{\omega}+I\left(  M;L\right)
_{\omega}-I\left(  A_{2}B^{n}T_{B};L\right)  _{\omega}\\
&  \ \ \ \ \ \ +H\left(  B^{n}|LM\right)  +H\left(  A_{2}T_{B}|B^{n}LM\right)
_{\omega}-H\left(  RS_{A}A_{2}T_{B}B^{n}|LM\right)  _{\omega}\\
&  =I\left(  ML;B^{n}\right)  _{\omega}+I\left(  ML;A_{2}T_{B}|B^{n}\right)
_{\omega}+I\left(  M;L\right)  _{\omega}-I\left(  A_{2}B^{n}T_{B};L\right)
_{\omega}\\
&  \ \ \ \ \ \ +H\left(  A_{2}T_{B}|B^{n}LM\right)  _{\omega}+I\left(
RS_{A}A_{2}T_{B}\rangle B^{n}LM\right)  _{\omega}.
\end{align*}
The first equality follows from evaluating the mutual information of the state
$\overline{\Phi}^{M\hat{M}}$ and the coherent information of the product state
$\Phi^{RB_{1}}\otimes\Phi^{S_{A}S_{B}}$. The first inequality follows from the
condition in (\ref{eq:+++_good-code}) and an application of the Alicki-Fannes'
inequality with $\delta^{\prime}$ vanishing when $\epsilon\rightarrow0$. We
suppress the term $n\delta^{\prime}$ in the following lines. The second
inequality follows from quantum data processing. The second equality follows
from applying the chain rule for quantum mutual information to $I(M;B^{n}%
A_{2}T_{B}L)_{\omega}$ and by expanding the coherent information $I\left(
RS_{A}\rangle B^{n}A_{2}T_{B}LM\right)  _{\omega}$. The third equality follows
from applying the chain rule for quantum mutual information to $I\left(
ML;B^{n}A_{2}T_{B}\right)  _{\omega}$ and from the definition of coherent
information. We continue below:%
\begin{align*}
&  =I\left(  ML;B^{n}\right)  _{\omega}+I\left(  RS_{A}A_{2}T_{B}\rangle
B^{n}LM\right)  _{\omega}\\
&  \ \ \ \ \ \ +I\left(  ML;A_{2}T_{B}|B^{n}\right)  _{\omega}+I\left(
M;L\right)  _{\omega}-I\left(  A_{2}B^{n}T_{B};L\right)  _{\omega}+H\left(
A_{2}T_{B}|B^{n}LM\right)  _{\omega}\\
&  =I\left(  ML;B^{n}\right)  _{\omega}+I\left(  RS_{A}A_{2}T_{B}\rangle
B^{n}LM\right)  _{\omega}\\
&  \ \ \ \ \ \ +H\left(  A_{2}T_{B}|B^{n}\right)  _{\omega}+I\left(
M;L\right)  _{\omega}-I\left(  A_{2}B^{n}T_{B};L\right)  _{\omega}\\
&  \leq I\left(  ML;B^{n}\right)  _{\omega}+I\left(  RS_{A}A_{2}T_{B}\rangle
B^{n}LM\right)  _{\omega}+n\left(  \tilde{C}+\tilde{Q}+\tilde{E}\right)  \\
&  =I\left(  X;B^{n}\right)  _{\omega}+I\left(  A\rangle B^{n}X\right)
_{\omega}+n\left(  \tilde{C}+\tilde{Q}+\tilde{E}\right)  .
\end{align*}
The first equality follows by rearranging terms. The second equality follows
by canceling terms. The inequality follows from subadditivity of the entropy
$H\left(  A_{2}T_{B}|B^{n}\right)  _{\omega}$, the fact that the entropy
$H\left(  A_{2}T_{B}|B^{n}\right)  _{\omega}$ can never be larger than the
logarithm of the dimension of the systems $A_{2}T_{B}$, that the mutual
information $I\left(  M;L\right)  _{\omega}$ can never be larger than the
logarithm of the dimension of the classical register $L$, and because
$I\left(  A_{2}B^{n}T_{B};L\right)  _{\omega}\geq0$. The last equality follows
from the definitions $A\equiv RT_{B}A_{2}S_{A}$ and $X\equiv ML$. Thus the
following inequality holds%
\[
n\left(  \bar{C}+\bar{Q}+\bar{E}\right)  \leq I\left(  X;B^{n}\right)
_{\omega}+I\left(  A\rangle B^{n}X\right)  _{\omega}+n\left(  \tilde{C}%
+\tilde{Q}+\tilde{E}\right)  +n\delta^{\prime},
\]
demonstrating that the inequality in (\ref{eq:CQE-bound}) applies to the net
rates. This concludes the catalytic proof of the converse theorem.

\subsection{Pure state ensembles are sufficient}

\label{sec:pure-suff}We prove that it is sufficient to consider an ensemble of
pure states as in the statement of Theorem~\ref{thm:main-theorem}\ rather than
an ensemble of mixed states as in (\ref{eq:converse-state}) in the proof of
our converse theorem. Our argument relies on a classic trick exploited in
quantum Shannon theory \cite{DHW05RI}. We first determine a spectral decomposition of the
mixed state ensemble, model the index of the pure states in the decomposition
as a classical variable $Y$, and then place this classical variable $Y$ in a
classical register. It follows that the communication rates can only improve,
and it is sufficient to consider an ensemble of pure states.

Consider that each mixed state in the ensemble in (\ref{eq:converse-state})
admits a spectral decomposition of the following form:%
\[
\rho_{x}^{AA^{\prime}}=\sum_{y}p\left(  y|x\right)  \psi_{x,y}^{AA^{\prime}}.
\]
We can thus represent the ensemble as follows:%
\begin{equation}
\rho^{XAB}\equiv\sum_{x,y}p(x)p\left(  y|x\right)  \left\vert x\right\rangle
\left\langle x\right\vert ^{X}\otimes\mathcal{N}^{A^{\prime}\rightarrow
B}(\psi_{x,y}^{AA^{\prime}}). \label{eq:non-isometric-state}%
\end{equation}
The inequalities in (\ref{eq:CQ-bound}-\ref{eq:CQE-bound}) for the dynamic
capacity region involve the mutual information $I(AX;B)_{\rho}$, the Holevo
information $I(X;B)_{\rho}$, and the coherent information $I(A\rangle
BX)_{\rho}$. As we show below, each of these entropic quantities can only
improve in each case if make the variable $y$ be part of the classical
variable. This improvement then implies that it is only necessary to consider
pure states in the dynamic capacity theorem.

Let $\theta^{XYAB}$ denote an augmented state of the following form:%
\begin{equation}
\theta^{XYAB}\equiv\sum_{x}p(x)p\left(  y|x\right)  \left\vert x\right\rangle
\left\langle x\right\vert ^{X}\otimes\left\vert y\right\rangle \left\langle
y\right\vert ^{Y}\otimes\mathcal{N}^{A^{\prime}\rightarrow B}(\psi
_{x,y}^{AA^{\prime}}). \label{eq:measured-state}%
\end{equation}
This state is actually a state of the form in (\ref{eq:main-theorem-state}) if
we subsume the classical variables $X$ and $Y$ into one classical variable.
The following three inequalities each follow from an application of the
quantum data processing inequality:%
\begin{align}
I(X;B)_{\rho}  &  =I(X;B)_{\theta}\leq I(XY;B)_{\theta},\\
I(AX;B)_{\rho}  &  =I(AX;B)_{\theta}\leq I(AXY;B)_{\theta}\\
I(A\rangle BX)_{\rho}  &  =I(A\rangle BX)_{\theta}\leq I(A\rangle
BXY)_{\theta}.
\end{align}
Each of these inequalities proves the desired result for the respective Holevo
information, mutual information, and coherent information, and it suffices to
consider an ensemble of pure states in Theorem~\ref{thm:main-theorem}.

\section{The quantum dynamic capacity formula}

\label{sec:dynamic-cap-formula}We introduce the quantum dynamic capacity
formula and show how additivity of it implies that the computation of the
Pareto optimal trade-off surface of the capacity region requires just a single channel use,
rather than an infinite number of them (as in regularized formulas). The
Pareto optimal trade-off surface consists of all points in the capacity region
that are Pareto optimal, in the sense that it is not possible to make
improvements in one resource without offsetting another resource (these are
essentially the boundary points of the region in our case). We then show how
several important capacity formulas in the quantum Shannon theory literature
are special cases of the quantum dynamic capacity formula.

\begin{definition}
[Quantum Dynamic Capacity Formula]The quantum dynamic capacity formula of a
quantum channel $\mathcal{N}$ is as follows:%
\begin{equation}
D_{\lambda,\mu}\left(  \mathcal{N}\right)  \equiv\max_{\sigma}I\left(
AX;B\right)  _{\sigma}+\lambda I\left(  A\rangle BX\right)  _{\sigma}%
+\mu\left(  I\left(  X;B\right)  _{\sigma}+I\left(  A\rangle BX\right)
_{\sigma}\right)  , \label{eq:objective}%
\end{equation}
where $\lambda,\mu\geq0$.
\end{definition}

\begin{definition}
The regularized quantum dynamic capacity formula is as follows:%
\[
D_{\lambda,\mu}^{\text{reg}}\left(  \mathcal{N}\right)  \equiv\lim
_{n\rightarrow\infty}\frac{1}{n}D_{\lambda,\mu}\left(  \mathcal{N}^{\otimes
n}\right)  .
\]

\end{definition}

\begin{lemma}
\label{thm:CEQ-single-letter}Suppose the quantum dynamic capacity formula is
additive for any two channels $\mathcal{N}$ and $\mathcal{M}$:%
\[
D_{\lambda,\mu}\left(  \mathcal{N\otimes M}\right)  = D_{\lambda,\mu}\left(
\mathcal{N}\right) +D_{\lambda,\mu}\left(
\mathcal{M}\right) .
\]
Then the regularized quantum dynamic capacity formula for $\mathcal{N}$ is equal to the quantum
dynamic capacity formula for $\mathcal{N}$:%
\[
D_{\lambda,\mu}^{\text{reg}}\left(  \mathcal{N}\right)  =D_{\lambda,\mu
}\left(  \mathcal{N}\right)  .
\]
In this sense, the regularized formula \textquotedblleft
single-letterizes\textquotedblright\ and it is not necessary to take the limit.
\end{lemma}

\begin{proof}
We prove the result using induction on $n$. The base case for $n=1$ is
trivial. Suppose the result holds for $n$: $D_{\lambda,\mu}(\mathcal{N}%
^{\otimes n})=nD_{\lambda,\mu}(\mathcal{N})$. Then the following chain of
equalities proves the inductive step:%
\begin{align*}
D_{\lambda,\mu}(\mathcal{N}^{\otimes n+1})  &  =D_{\lambda,\mu}(\mathcal{N}%
\otimes\mathcal{N}^{\otimes n})\\
&  =D_{\lambda,\mu}(\mathcal{N})+D_{\lambda,\mu}(\mathcal{N}^{\otimes n})\\
&  =D_{\lambda,\mu}(\mathcal{N})+nD_{\lambda,\mu}(\mathcal{N}).
\end{align*}
The first equality follows by expanding the tensor product. The second
critical equality follows from the assumption that the formula is additive.
The final equality follows from the induction hypothesis.
\end{proof}

\begin{theorem}
Single-letterization of the quantum dynamic
capacity formula implies that the computation of the Pareto optimal trade-off
surface of the quantum dynamic capacity region requires an optimization over a
single channel use.
\end{theorem}

\begin{proof}
We employ ideas from Ref.~\cite{BV04} for the proof.
We would like to characterize all the points in the capacity region that are
Pareto optimal. Such a task is standard vector optimization in the theory of
Pareto trade-off analysis (see Section~4.7 of Ref.~\cite{BV04}). We can phrase
the optimization task as the following scalarization of the vector
optimization task:%
\begin{equation}
\max_{C,Q,E,p\left(  x\right)  ,\phi_{x}}w_{C}C+w_{Q}Q+w_{E}E
\label{eq:scalarization}%
\end{equation}
subject to%
\begin{align}
C+2Q  &  \leq I(AX;B^{n})_{\sigma},\label{eq:cond-1}\\
Q+E  &  \leq I(A\rangle B^{n}X)_{\sigma},\label{eq:cond-2}\\
C+Q+E  &  \leq I(X;B^{n})_{\sigma}+I(A\rangle B^{n}X)_{\sigma},
\label{eq:cond-3}%
\end{align}
where the maximization is over all $C$, $Q$, and $E$ and over probability
distributions $p_{X}\left(  x\right)  $ and bipartite states $\phi
_{x}^{AA^{\prime n}}$. The geometric interpretation of the scalarization task is
that we are trying to find a supporting plane of the dynamic capacity region
where the weight vector $\left(  w_{C},w_{Q},w_{E}\right)  $ is the normal
vector of the plane and the value of its inner product with $\left(
C,Q,E\right)  $ characterizes the offset of the plane.

The Lagrangian of the above optimization problem is%
\begin{align*}
\mathcal{L}\left(  C,Q,E,p_{X}\left(  x\right)  ,\phi_{x}^{AA^{\prime n}%
},\lambda_{1},\lambda_{2},\lambda_{3}\right)   &  \equiv w_{C}C+w_{Q}%
Q+w_{E}E+\lambda_{1}\left(  I\left(  AX;B^{n}\right)  _{\sigma}-\left(
C+2Q\right)  \right) \\
&  \ \ \ \ \ \ +\lambda_{2}\left(  I\left(  A\rangle B^{n}X\right)  _{\sigma
}-\left(  Q+E\right)  \right) \\
&  \ \ \ \ \ \ \ +\lambda_{3}\left(  I\left(  X;B^{n}\right)  _{\sigma
}+I\left(  A\rangle B^{n}X\right)  _{\sigma}-\left(  C+Q+E\right)  \right)  ,
\end{align*}
and the Lagrange dual function $g$ \cite{BV04} is%
\[
g\left(  \lambda_{1},\lambda_{2},\lambda_{3}\right)  \equiv\sup
_{C,Q,E,p\left(  x\right)  ,\phi_{x}^{AA^{\prime n}}}\mathcal{L}\left(
C,Q,E,p_{X}\left(  x\right)  ,\phi_{x}^{AA^{\prime n}},\lambda_{1},\lambda
_{2},\lambda_{3}\right)  ,
\]
where $\lambda_{1},\lambda_{2},\lambda_{3}\geq0$. The optimization task simplifies
 if  the Lagrange dual function does. Thus, we
rewrite the Lagrange dual function as follows:%
\begin{align*}
g\left(  \lambda_{1},\lambda_{2},\lambda_{3}\right)   &  =\sup_{C,Q,E,p\left(
x\right)  ,\phi_{x}^{AA^{\prime n}}}w_{C}C+w_{Q}Q+w_{E}E+\lambda_{1}\left(
I\left(  AX;B^{n}\right)  _{\sigma}-\left(  C+2Q\right)  \right) \\
&  \ \ \ \ \ \ \ +\lambda_{2}\left(  I\left(  A\rangle B^{n}X\right)
_{\sigma}-\left(  Q+E\right)  \right) \\
&  \ \ \ \ \ \ \ +\lambda_{3}\left(  I\left(  X;B^{n}\right)  _{\sigma
}+I\left(  A\rangle B^{n}X\right)  _{\sigma}-\left(  C+Q+E\right)  \right) \\
&  =\sup_{C,Q,E,p\left(  x\right)  ,\phi_{x}^{AA^{\prime n}}}\left(
w_{C}-\lambda_{1}-\lambda_{3}\right)  C+\left(  w_{Q}-2\lambda_{1}-\lambda
_{2}-\lambda_{3}\right)  Q+\left(  w_{E}-\lambda_{2}-\lambda_{3}\right)  E\\
&  \ \ \ \ \ \ \ +\lambda_{1}\left(  I\left(  AX;B^{n}\right)  _{\sigma}%
+\frac{\lambda_{2}}{\lambda_{1}}I\left(  A\rangle B^{n}X\right)  _{\sigma
}+\frac{\lambda_{3}}{\lambda_{1}}\left(  I\left(  X;B^{n}\right)  _{\sigma
}+I\left(  A\rangle B^{n}X\right)  _{\sigma}\right)  \right) \\
&  =\sup_{C,Q,E}\left(  w_{C}-\lambda_{1}-\lambda_{3}\right)  C+\left(
w_{Q}-2\lambda_{1}-\lambda_{2}-\lambda_{3}\right)  Q+\left(  w_{E}-\lambda
_{2}-\lambda_{3}\right)  E\\
&  \ \ \ \ \ \ \ +\lambda_{1}\left(  \max_{p\left(  x\right)  ,\phi
_{x}^{AA^{\prime n}}}I\left(  AX;B^{n}\right)  _{\sigma}+\frac{\lambda_{2}%
}{\lambda_{1}}I\left(  A\rangle B^{n}X\right)  _{\sigma}+\frac{\lambda_{3}%
}{\lambda_{1}}\left(  I\left(  X;B^{n}\right)  _{\sigma}+I\left(  A\rangle
B^{n}X\right)  _{\sigma}\right)  \right)  .
\end{align*}
The first equality follows by definition. The second equality follows from
some algebra, and the last follows because the Lagrange dual function factors
into two separate optimization tasks:\ one over $C$, $Q$, and $E$ and another
that is equivalent to the quantum dynamic capacity formula with $\lambda
=\lambda_{2}/\lambda_{1}$ and $\mu=\lambda_{3}/\lambda_{1}$.
Thus, the
computation of the Pareto optimal trade-off surface requires just a single use
of the channel if the quantum dynamic capacity formula in
(\ref{eq:objective}) single-letterizes.
\end{proof}

\subsection{Special cases of the quantum dynamic capacity formula}

We now show how several capacity formulas of a quantum channel, including the
entanglement-assisted classical capacity~\cite{BSST01}, the Lloyd-Shor-Devetak
(LSD) formula for the quantum capacity~\cite{Lloyd96,Shor02,Devetak03}, and
the Holevo-Schumacher-Westmoreland (HSW) formula for the classical
capacity~\cite{Hol98,SW97} are special cases of the quantum dynamic capacity formula.

We first give a geometric interpretation of these special cases before
proceeding to the proofs. Recall that the dynamic capacity region has the
simple interpretation as a translation of the three-faced unit resource
capacity region along the classically-enhanced father trade-off curve (see
Figure~\ref{fig:full-plot-dephasing} for the example of the region of the
dephasing channel). Any particular weight vector $\left(  w_{C},w_{Q}%
,w_{E}\right)  $ in (\ref{eq:scalarization}) gives a set of parallel planes
that slice through the $\left(  C,Q,E\right)  $\ space, and the goal of the
scalar optimization task is to find one of these planes that is a supporting
plane, intersecting a point (or a set of points) on the trade-off surface of
the dynamic capacity region. We consider three special planes:

\begin{enumerate}
\item The first corresponds to the plane containing the vectors of super-dense
coding and teleportation. The normal vector of this plane is $(1,2,0)$, and
suppose that we set the weight vector in (\ref{eq:scalarization}) to be this
vector. Then the optimization program finds the set of points on the trade-off
surface such that a plane with this normal vector is a supporting plane for
the region. The optimization program singles out (\ref{eq:cond-1}), and we can
think of this as being equivalent to setting $\lambda_{2},\lambda_{3}=0$ in
the Lagrange dual function. We show below that the optimization program
becomes equivalent to finding the entanglement-assisted capacity~\cite{BSST01}%
, in the sense that the quantum dynamic capacity formula becomes the
entanglement-assisted capacity formula.

\item The next plane contains the vectors of teleportation and entanglement
distribution. The normal vector of this plane is $\left(  0,1,1\right)  $.
Setting the weight vector in (\ref{eq:scalarization}) to be this vector makes
the optimization program single out (\ref{eq:cond-2}), and we can think of
this as being equivalent to setting $\lambda_{1},\lambda_{3}=0$ in the
Lagrange dual function. We show below that the optimization program becomes
equivalent to finding the quantum capacity~\cite{Lloyd96,Shor02,Devetak03}, in
the sense that the quantum dynamic capacity formula becomes the LSD formula
for the quantum capacity.

\item A third plane contains the vectors of super-dense coding and
entanglement distribution. The normal vector of this plane is $\left(
1,1,1\right)  $. Setting the weight vector in (\ref{eq:scalarization}) to be
this vector makes the optimization program single out (\ref{eq:cond-3}), and
we can think of this as being equivalent to setting $\lambda_{1},\lambda
_{2}=0$ in the Lagrange dual function. We show below that the optimization
becomes equivalent to finding the classical
capacity~\cite{Lloyd96,Shor02,Devetak03}, in the sense that the quantum
dynamic capacity formula becomes the HSW\ formula for the classical capacity.
\end{enumerate}

\begin{corollary}
The quantum dynamic capacity formula is equivalent to the
entanglement-assisted classical capacity formula when $\lambda,\mu=0$, in the
sense that%
\[
\max_{\sigma}I\left(  AX;B\right)  =\max_{\phi^{AA^{\prime}}}I\left(
A;B\right)  .
\]

\end{corollary}

\begin{proof}
The inequality $\max_{\sigma}I\left(  AX;B\right)  \geq\max_{\phi^{AA^{\prime
}}}I\left(  A;B\right)  $ follows because the state $\sigma$ is of the form in
(\ref{eq:maximization-state}) and we can always choose $p_{X}\left(  x\right)
=\delta_{x,x_{0}}$ and $\phi_{x_{0}}^{AA^{\prime}}$ to be the state that
maximizes $I\left(  A;B\right)  $.

We now show the other inequality $\max_{\sigma}I\left(  AX;B\right)  \leq
\max_{\phi^{AA^{\prime}}}I\left(  A;B\right)  $. First, consider that the
following chain of equalities holds for any state $\phi^{ABE}$ resulting from
the isometric extension of the channel:%
\begin{align*}
I\left(  A;B\right)   &  =H\left(  B\right)  +H\left(  A\right)  -H\left(
AB\right) \\
&  =H\left(  B\right)  +H\left(  BE\right)  -H\left(  E\right) \\
&  =H\left(  B\right)  +H\left(  B|E\right)  .
\end{align*}
In this way, we see that the mutual information is purely a function of the
channel input density operator Tr$_{A}\left\{  \phi^{AA^{\prime}}\right\}  $.
Then consider any state $\sigma$ of the form in (\ref{eq:maximization-state}).
The following chain of inequalities holds%
\begin{align*}
I\left(  AX;B\right)  _{\sigma}  &  =H\left(  A|X\right)  _{\sigma}+H\left(
B\right)  _{\sigma}-H\left(  E|X\right)  _{\sigma}\\
&  =H\left(  BE|X\right)  _{\sigma}+H\left(  B\right)  _{\sigma}-H\left(
E|X\right)  _{\sigma}\\
&  =H\left(  B|EX\right)  _{\sigma}+H\left(  B\right)  _{\sigma}\\
&  \leq H\left(  B|E\right)  _{\sigma}+H\left(  B\right)  _{\sigma}\\
&  \leq\max_{\phi^{AA^{\prime}}}I\left(  A;B\right)  .
\end{align*}
The first equality follows by expanding the mutual information. The second
equality follows because the state on $ABE$ is pure when conditioned on $X$.
The third equality follows from the entropy chain rule. The first inequality
follows from strong subadditivity, and the last follows because the state
after tracing out systems $X$ and $A$ is a particular state that arises from
the channel and cannot be larger than the maximum.
\end{proof}

\begin{corollary}
The quantum dynamic capacity formula is equivalent to the LSD\ quantum
capacity formula in the limit where $\lambda\rightarrow\infty$ and $\mu$ is
fixed, in the sense that%
\[
\max_{\sigma}I\left(  A\rangle BX\right)  =\max_{\phi^{AA^{\prime}}}I\left(
A\rangle B\right)  .
\]

\end{corollary}

\begin{proof}
The inequality $\max_{\sigma}I\left(  A\rangle BX\right)  \geq\max
_{\phi^{AA^{\prime}}}I\left(  A\rangle B\right)  $ follows because the state
$\sigma$ is of the form in (\ref{eq:maximization-state}) and we can always
choose $p_{X}\left(  x\right)  =\delta_{x,x_{0}}$ and $\phi_{x_{0}%
}^{AA^{\prime}}$ to be the state that maximizes $I\left(  A\rangle B\right)  $.

The inequality $\max_{\sigma}I\left(  A\rangle BX\right)  \leq\max
_{\phi^{AA^{\prime}}}I\left(  A\rangle B\right)  $ follows because $I\left(
A\rangle BX\right)  =\sum_{x}p_{X}\left(  x\right)  I\left(  A\rangle
B\right)  _{\phi_{x}}$ and the maximum is always greater than the average.
\end{proof}

\begin{corollary}
The quantum dynamic capacity formula is equivalent to the HSW\ classical
capacity formula in the limit where $\mu\rightarrow\infty$ and $\lambda$ is
fixed, in the sense that%
\[
\max_{\sigma}I\left(  A\rangle BX\right)  _{\sigma}+I\left(  X;B\right)
_{\sigma}=\max_{\left\{  p_{X}\left(  x\right)  ,\psi_{x}\right\}  }I\left(
X;B\right)  .
\]

\end{corollary}

The inequality $\max_{\sigma}I\left(  A\rangle BX\right)  _{\sigma}+I\left(
X;B\right)  _{\sigma}\geq\max_{\left\{  p_{X}\left(  x\right)  ,\psi
_{x}\right\}  }I\left(  X;B\right)  $ follows by choosing $\sigma$ to be the
pure ensemble that maximizes $I\left(  X;B\right)  $ and noting that $I\left(
A\rangle BX\right)  _{\sigma}$ vanishes for a pure ensemble.

We now prove the inequality $\max_{\sigma}I\left(  A\rangle BX\right)
_{\sigma}+I\left(  X;B\right)  _{\sigma}\leq\max_{\left\{  p_{X}\left(
x\right)  ,\psi_{x}\right\}  }I\left(  X;B\right)  $. Consider a state
$\omega^{XYBE}$ obtained by performing a von Neumann measurement on the $A$
system of the state $\sigma^{XABE}$. Then%
\begin{align*}
I\left(  A\rangle BX\right)  _{\sigma}+I\left(  X;B\right)  _{\sigma}  &
=H\left(  B\right)  _{\sigma}-H\left(  E|X\right)  _{\sigma}\\
&  =H\left(  B\right)  _{\omega}-H\left(  E|X\right)  _{\omega}\\
&  \leq H\left(  B\right)  _{\omega}-H\left(  E|XY\right)  _{\omega}\\
&  =H\left(  B\right)  _{\omega}-H\left(  B|XY\right)  _{\omega}\\
&  =I\left(  XY;B\right)  _{\omega}\\
&  \leq\max_{\left\{  p_{X}\left(  x\right)  ,\psi_{x}\right\}  }I\left(
X;B\right)  .
\end{align*}
The first equality follows by expanding the conditional coherent information
and the Holevo information. The second equality follows because the measured
$A$ system is not involved in the entropies. The first inequality follows
because conditioning does not increase entropy. The third equality follows
because the state $\omega$ is pure when conditioned on $X$ and $Y$. The fourth
equality follows by definition, and the last inequality follows for clear reasons.

\section{Single-letter dynamic capacity region for the quantum Hadamard
channels}

\label{sec:single-letter-hadamard}Below we show that the regularization in
(\ref{eq:multi-letter}) is not necessary if the quantum channel is a Hadamard
channel. This result holds because a Hadamard channel has a special structure.
The development of the proof is similar to that in Ref.~\cite{BHTW10}, but
simplified because we obtain the single-letter result more directly.

\begin{theorem}
\label{thm:Hadamard-theorem}The dynamic capacity region $\mathcal{C}%
_{\mathrm{{CQE}}}(\mathcal{N}_{\text{H}})$ of a quantum Hadamard channel
$\mathcal{N}_{\text{H}}$ is equal to its one-shot region $\mathcal{C}%
_{\mathrm{{CQE}}}^{(1)}(\mathcal{N}_{\text{H}})$.
\end{theorem}

The proof of the above theorem follows in two parts:\ 1)\ the below lemma
shows the quantum dynamic capacity formula is additive when one of the
channels is Hadamard and 2)\ the induction argument in
Lemma~\ref{thm:CEQ-single-letter}\ that proves single-letterization.

\begin{lemma}
\label{lem:CEQ-base-case}The following additivity relation holds for a
Hadamard channel $\mathcal{N}_{H}$ and any other channel $\mathcal{N}$:%
\[
D_{\lambda,\mu}(\mathcal{N}_{H}\otimes\mathcal{N})=D_{\lambda,\mu}%
(\mathcal{N}_{H})+D_{\lambda,\mu}(\mathcal{N}).
\]

\end{lemma}

\begin{proof}
We first note that the inequality $D_{\lambda,\mu}(\mathcal{N}_{H}%
\otimes\mathcal{N})\geq D_{\lambda,\mu}(\mathcal{N}_{H})+D_{\lambda,\mu
}(\mathcal{N})$ holds for any two channels simply by selecting the state
$\sigma$ in the maximization to be a tensor product of the ones that
individually maximize $D_{\lambda,\mu}(\mathcal{N}_{H})$ and $D_{\lambda,\mu
}(\mathcal{N})$.

So we prove that the non-trivial inequality $D_{\lambda,\mu}(\mathcal{N}%
_{H}\otimes\mathcal{N})\leq D_{\lambda,\mu}(\mathcal{N}_{H})+D_{\lambda,\mu
}(\mathcal{N})$ holds when the first channel is a Hadamard channel. Since the
first channel is Hadamard, it is degradable and its degrading map has a
particular structure: there are maps $\mathcal{D}_{1}^{B_{1}\rightarrow Y}$
and $\mathcal{D}_{2}^{Y\rightarrow E_{1}}$ where $Y$ is a classical register
and such that the degrading map is $\mathcal{D}_{2}^{Y\rightarrow E_{1}}%
\circ\mathcal{D}_{1}^{B_{1}\rightarrow Y}$~\cite{BHTW10,KMNR07}. Suppose the
state we are considering to input to the tensor product channel is%
\[
\rho^{XAA_{1}^{\prime}A_{2}^{\prime}}\equiv\sum_{x}p_{X}\left(  x\right)
\left\vert x\right\rangle \left\langle x\right\vert ^{X}\otimes\phi
_{x}^{AA_{1}^{\prime}A_{2}^{\prime}},
\]
and this state is the one that maximizes $D_{\lambda,\mu}(\mathcal{N}%
_{H}\otimes\mathcal{N})$. Suppose that the output of the first channel is%
\[
\theta^{XAB_{1}E_{1}A_{2}^{\prime}}\equiv U_{\mathcal{N}_{H}}^{A_{1}^{\prime
}\rightarrow B_{1}E_{1}}(\rho^{XAA_{1}^{\prime}A_{2}^{\prime}}),
\]
and the output of the second channel is%
\[
\omega^{XAB_{1}E_{1}B_{2}E_{2}}\equiv U_{\mathcal{N}}^{A_{2}^{\prime
}\rightarrow B_{2}E_{2}}(\theta^{XAB_{1}E_{1}A_{2}^{\prime}}).
\]
Finally, we define the following state as the result of applying the first
part of the Hadamard degrading map (a von Neumann measurement) to $\omega$:%
\[
\sigma^{XYAE_{1}B_{2}E_{2}}\equiv\mathcal{D}_{1}^{B_{1}\rightarrow Y}%
(\omega^{XAB_{1}E_{1}B_{2}E_{2}}).
\]
In particular, the state $\sigma$ on systems $AE_{1}B_{2}E_{2}$ is pure when
conditioned on $X$ and $Y$. Then the following chain of inequalities holds:%
\begin{align*}
D_{\lambda,\mu}\left(  \mathcal{N}_{H}\otimes\mathcal{N}\right)   &  =I\left(
AX;B_{1}B_{2}\right)  _{\omega}+\lambda I\left(  A\rangle B_{1}B_{2}X\right)
_{\omega}+\mu\left(  I\left(  X;B_{1}B_{2}\right)  _{\omega}+I\left(  A\rangle
B_{1}B_{2}X\right)  _{\omega}\right) \\
&  =H\left(  B_{1}B_{2}E_{1}E_{2}|X\right)  _{\omega}+\lambda H\left(
B_{1}B_{2}|X\right)  _{\omega}+\left(  \mu+1\right)  H\left(  B_{1}%
B_{2}\right)  _{\omega}-\left(  \lambda+\mu+1\right)  H\left(  E_{1}%
E_{2}|X\right)  _{\omega}\\
&  =H\left(  B_{1}E_{1}|X\right)  _{\omega}+\lambda H\left(  B_{1}|X\right)
_{\omega}+\left(  \mu+1\right)  H\left(  B_{1}\right)  _{\omega}-\left(
\lambda+\mu+1\right)  H\left(  E_{1}|X\right)  _{\omega}+\\
&  \ \ \ \ \ H\left(  B_{2}E_{2}|B_{1}E_{1}X\right)  _{\omega}+\lambda
H\left(  B_{2}|B_{1}X\right)  _{\omega}+\left(  \mu+1\right)  H\left(
B_{2}|B_{1}\right)  _{\omega}-\left(  \lambda+\mu+1\right)  H\left(
E_{2}|E_{1}X\right)  _{\omega}\\
&  \leq H\left(  B_{1}E_{1}|X\right)  _{\theta}+\lambda H\left(
B_{1}|X\right)  _{\theta}+\left(  \mu+1\right)  H\left(  B_{1}\right)
_{\theta}-\left(  \lambda+\mu+1\right)  H\left(  E_{1}|X\right)  _{\theta}+\\
&  \ \ \ \ \ H\left(  B_{2}E_{2}|YX\right)  _{\sigma}+\lambda H\left(
B_{2}|YX\right)  _{\sigma}+\left(  \mu+1\right)  H\left(  B_{2}\right)
_{\sigma}-\left(  \lambda+\mu+1\right)  H\left(  E_{2}|YX\right)  _{\sigma}\\
&  =I\left(  AA_{2}^{\prime}X;B_{1}\right)  _{\theta}+\lambda I\left(
AA_{2}^{\prime}\rangle B_{1}X\right)  _{\theta}+\mu\left(  I\left(
X;B_{1}\right)  _{\theta}+I\left(  AA_{2}^{\prime}\rangle B_{1}X\right)
_{\theta}\right)  +\\
&  \ \ \ \ \ I\left(  AE_{1}YX;B_{2}\right)  _{\sigma}+\lambda I\left(
AE_{1}\rangle B_{2}YX\right)  _{\sigma}+\mu\left(  I\left(  YX;B_{2}\right)
_{\sigma}+I\left(  AE_{1}\rangle B_{2}YX\right)  _{\sigma}\right) \\
&  \leq D_{\lambda,\mu}\left(  \mathcal{N}_{H}\right)  +D_{\lambda,\mu}\left(
\mathcal{N}\right)  .
\end{align*}
The first equality follows by evaluating the quantum dynamic capacity formula
$D_{\lambda,\mu}\left(  \mathcal{N}_{H}\otimes\mathcal{N}\right)  $ on the
state $\rho$. The next two equalities follow by rearranging entropies and
because the state $\omega$ on systems $AB_{1}E_{1}B_{2}E_{2}$ is pure when
conditioned on $X$. The inequality in the middle is the crucial one and
follows from the Hadamard structure of the channel: we exploit monotonicity of
conditional entropy under quantum operations so that $H\left(  B_{2}%
|B_{1}X\right)  _{\omega}\leq H\left(  B_{2}|YX\right)  _{\sigma}$ and
$H\left(  E_{2}|YX\right)  _{\sigma}\leq H\left(  E_{2}|E_{1}X\right)
_{\omega}$. The next equality follows by rearranging entropies and the final
one follows because $\theta$ is a state of the form
(\ref{eq:maximization-state}) for the first channel while $\sigma$ is a state
of the form (\ref{eq:maximization-state}) for the second channel.
\end{proof}

\section{The dynamic capacity region of a dephasing channel}

\label{sec:dephasing}The below theorem shows that the full dynamic capacity
region admits a particularly simple form when the noisy quantum channel is a
qubit dephasing channel $\Delta_{p}$ where%
\begin{align*}
\Delta_{p}\left(  \rho\right)   &  \equiv\left(  1-p\right)  \rho
+p\Delta\left(  \rho\right)  ,\\
\Delta\left(  \rho\right)   &  \equiv\left\langle 0\left\vert \rho\right\vert
0\right\rangle \left\vert 0\right\rangle \left\langle 0\right\vert
+\left\langle 1\left\vert \rho\right\vert 1\right\rangle \left\vert
1\right\rangle \left\langle 1\right\vert .
\end{align*}
A dephasing channel is an example of a quantum Hadamard channel~\cite{BHTW10}%
.\footnote{Br\'{a}dler showed that cloning channels and an Unruh channel are
also in the Hadamard class \cite{B09}.} Figure~\ref{fig:full-plot-dephasing}%
\ plots this region for the case of a dephasing channel with dephasing
parameter $p=0.2$.%
\begin{figure}
[ptb]
\begin{center}
\includegraphics[
natheight=5.066900in,
natwidth=6.746400in,
height=3.0433in,
width=4.0413in
]%
{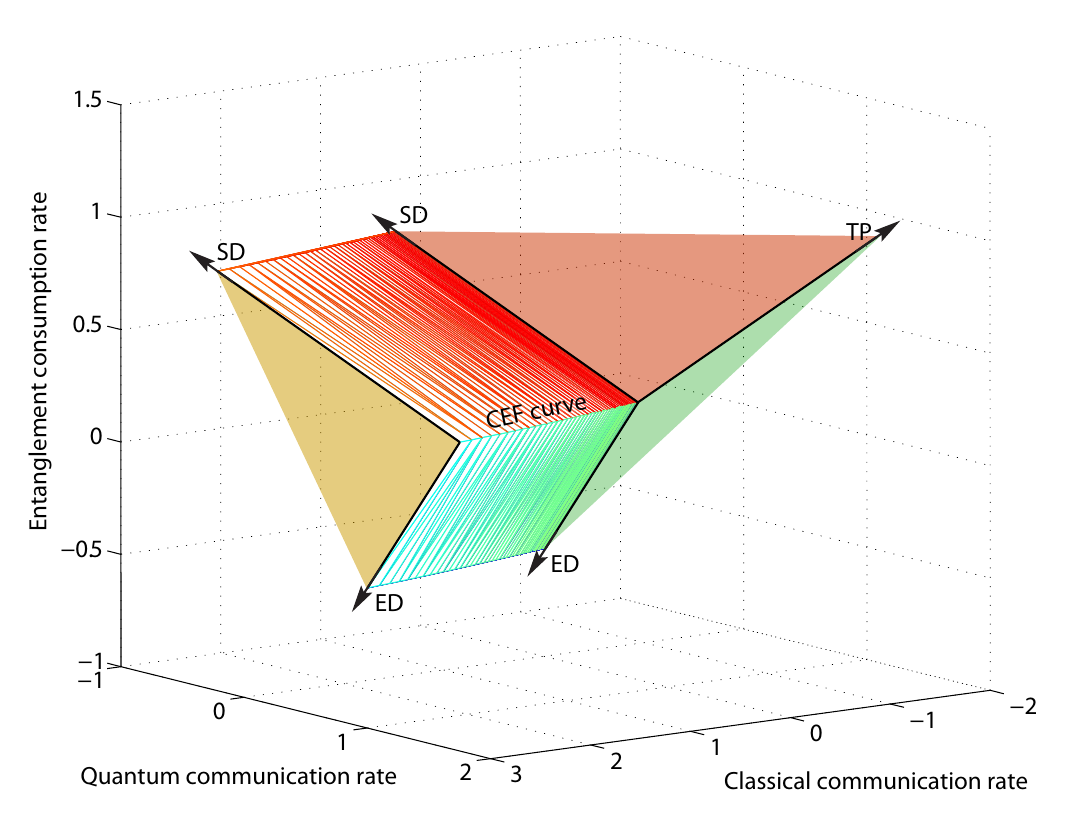}%
\caption{(Color online) A plot of the dynamic capacity region for a qubit
dephasing channel with dephasing parameter $p=0.2$. The plot shows that the
classically-enhanced father (CEF) trade-off curve lies along the boundary of
the dynamic capacity region. The rest of the region is simply the combination
of the CEF points with the unit protocols teleportation (TP), super-dense
coding (SD), and entanglement distribution (ED).}%
\label{fig:full-plot-dephasing}%
\end{center}
\end{figure}
The proof of the following theorem exploits the same techniques as in
Ref.~\cite{BHTW10}.

\begin{theorem}
\label{thm:dephasing}The dynamic capacity region $\mathcal{C}_{\text{CQE}%
}(\Delta_{p})$ of a dephasing channel with dephasing parameter $p$\ is the set
of all $C$, $Q$, and $E$ such that%
\begin{align}
C+2Q  &  \leq1+H_{2}\left(  \nu\right)  -H_{2}(\gamma\left(  \nu,p\right)
),\\
Q+E  &  \leq H_{2}\left(  \nu\right)  -H_{2}(\gamma\left(  \nu,p\right)  ),\\
C+Q+E  &  \leq1-H_{2}(\gamma\left(  \nu,p\right)  ),
\end{align}
where $\nu\in\left[  0,1/2\right]  $, $H_{2}$ is the binary entropy function,
and%
\[
\gamma\left(  \nu,p\right)  \equiv\frac{1}{2}+\frac{1}{2}\sqrt{1-16\cdot
\frac{p}{2}\left(  1-\frac{p}{2}\right)  \nu(1-\nu)}.
\]

\end{theorem}

\begin{proof}
We first notice that it suffices to consider an ensemble of pure states whose
reductions to $A^{\prime}$ are diagonal in the dephasing basis (as in Lemma~11
of Ref.~\cite{BHTW10}). Next we prove below that it is sufficient to consider
an ensemble of the following form to characterize the boundary points of the
region:%
\begin{equation}
\frac{1}{2}\left\vert 0\right\rangle \left\langle 0\right\vert ^{X}\otimes
\psi_{0}^{AA^{\prime}}+\frac{1}{2}\left\vert 1\right\rangle \left\langle
1\right\vert ^{X}\otimes\psi_{1}^{AA^{\prime}}, \label{eq:mu-cq-state-CEQ}%
\end{equation}
where $\psi_{0}^{AA^{\prime}}$ and $\psi_{1}^{AA^{\prime}}$ are pure states,
defined as follows for $\nu\in\left[  0,1/2\right]  $:%
\begin{align}
\text{Tr}_{A}\left\{  \psi_{0}^{AA^{\prime}}\right\}   &  =\nu\left\vert
0\right\rangle \left\langle 0\right\vert ^{A^{\prime}}+\left(  1-\nu\right)
\left\vert 1\right\rangle \left\langle 1\right\vert ^{A^{\prime}%
},\label{eq:1st-mu-state-CEQ}\\
\text{Tr}_{A}\left\{  \psi_{1}^{AA^{\prime}}\right\}   &  =\left(
1-\nu\right)  \left\vert 0\right\rangle \left\langle 0\right\vert ^{A^{\prime
}}+\nu\left\vert 1\right\rangle \left\langle 1\right\vert ^{A^{\prime}}.
\label{eq:2nd-mu-state-CEQ}%
\end{align}
We now prove the above claim. We assume without loss of generality that the
dephasing basis is the computational basis. Consider a classical-quantum state
with a finite number $N$ of conditional density operators $\phi_{x}%
^{AA^{\prime}}$ whose reduction to $A^{\prime}$ is diagonal:%
\begin{equation}
\rho^{XAA^{\prime}}\equiv\sum_{x=0}^{N-1}p_{X}\left(  x\right)  |x\rangle
\langle x|^{X}\otimes\phi_{x}^{AA^{\prime}}.\nonumber
\end{equation}
We can form a new classical-quantum state with double the number of
conditional density operators by \textquotedblleft
bit-flipping\textquotedblright\ the original conditional density operators:%
\[
\sigma^{XAA^{\prime}}\equiv\frac{1}{2}\sum_{x=0}^{N-1}p_{X}\left(  x\right)
\left(  |x\rangle\langle x|^{X}\otimes\phi_{x}^{AA^{\prime}}+|x+N\rangle
\langle x+N|^{X}\otimes X^{A^{\prime}}\phi_{x}^{AA^{\prime}}X^{A^{\prime}%
}\right)  ,
\]
where $X$ is the $\sigma_{X}$ \textquotedblleft bit-flip\textquotedblright%
\ Pauli operator. Consider the following chain of inequalities that holds for
all $\lambda,\mu\geq0$:%
\begin{align*}
&  I\left(  AX;B\right)  _{\rho}+\lambda I\left(  A\rangle BX\right)  _{\rho
}+\mu\left(  I\left(  X;B\right)  _{\rho}+I\left(  A\rangle BX\right)  _{\rho
}\right) \\
&  =H\left(  A|X\right)  _{\rho}+\left(  \mu+1\right)  H\left(  B\right)
_{\rho}+\lambda H\left(  B|X\right)  _{\rho}-\left(  \lambda+\mu+1\right)
H\left(  E|X\right)  _{\rho}\\
&  \leq\left(  \mu+1\right)  H\left(  B\right)  _{\sigma}+H\left(  A|X\right)
_{\sigma}+\lambda H\left(  B|X\right)  _{\sigma}-\left(  \lambda+\mu+1\right)
H\left(  E|X\right)  _{\sigma}\\
&  =\left(  \mu+1\right)  +H\left(  A|X\right)  _{\sigma}+\lambda H\left(
B|X\right)  _{\sigma}-\left(  \lambda+\mu+1\right)  H\left(  E|X\right)
_{\sigma}\\
&  =\left(  \mu+1\right)  +\sum_{x}p_{X}\left(  x\right)  \left[  H\left(
A\right)  _{\phi_{x}}+\lambda H\left(  B\right)  _{\phi_{x}}-\left(
\lambda+\mu+1\right)  H\left(  E\right)  _{\phi_{x}}\right] \\
&  \leq\left(  \mu+1\right)  +\max_{x}\left[  H\left(  A\right)  _{\phi_{x}%
}+\lambda H\left(  B\right)  _{\phi_{x}}-\left(  \lambda+\mu+1\right)
H\left(  E\right)  _{\phi_{x}}\right] \\
&  =\left(  \mu+1\right)  +H\left(  A\right)  _{\phi_{x}^{\ast}}+\lambda
H\left(  B\right)  _{\phi_{x}^{\ast}}-\left(  \lambda+\mu+1\right)  H\left(
E\right)  _{\phi_{x}^{\ast}}.
\end{align*}
The first equality follows by standard entropic manipulations. The second
equality follows because the conditional entropy $H\left(  B|X\right)  $ is
invariant under a bit-flipping unitary on the input state that commutes with
the channel: $H(B)_{X\rho_{x}^{B}X}=H(B)_{\rho_{x}^{B}}$. Furthermore, a bit
flip on the input state does not change the eigenvalues for the output of the
dephasing channel's complementary channel: $H(E)_{\mathcal{N}^{c}(X\rho
_{x}^{A^{\prime}}X)}=H(E)_{\mathcal{N}^{c}(\rho_{x}^{A^{\prime}})}$. The first
inequality follows because entropy is concave, i.e., the local state
$\sigma^{B}$ is a mixed version of $\rho^{B}$. The third equality follows
because $H(B)_{\sigma^{B}}=H\left(  \sum_{x}\frac{1}{2}p_{X}\left(  x\right)
(\rho_{x}^{B}+X\rho_{x}^{B}X)\right)  =H\left(  \frac{1}{2}\sum_{x}%
p_{X}\left(  x\right)  I\right)  =1$. The fourth equality follows because the
system $X$ is classical. The second inequality follows because the maximum
value of a realization of a random variable is not less than its expectation.
The final equality simply follows by defining $\phi_{x}^{\ast}$ to be the
conditional density operator on systems $A$, $B$, and $E$ that arises from
sending through the channel a state whose reduction to $A^{\prime}$ is of the
form $\nu\left\vert 0\right\rangle \left\langle 0\right\vert ^{A^{\prime}%
}+\left(  1-\nu\right)  \left\vert 1\right\rangle \left\langle 1\right\vert
^{A^{\prime}}$. Thus, an ensemble of the kind in (\ref{eq:mu-cq-state-CEQ}) is
sufficient to attain a point on the boundary of the region.

Evaluating the entropic quantities in Theorem~\ref{thm:main-theorem}\ on a
state of the above form then gives the expression for the region in
Theorem~\ref{thm:dephasing}.
\end{proof}

\section{Single-letter dynamic capacity region for the quantum erasure
channels}

\label{sec:single-letter-erasure}Below we show that the regularization in
(\ref{eq:multi-letter}) is not necessary if the quantum channel is a quantum
erasure channel. The quantum erasure channel also has a special structure, but
the proof proceeds differently from that for a quantum Hadamard channel.

A quantum erasure channel with erasure parameter $\epsilon$\ is the following
map \cite{GBP97}:%
\[
\mathcal{N}_{\epsilon}\left(  \rho\right)  \equiv\left(  1-\epsilon\right)
\rho+\epsilon\left\vert e\right\rangle \left\langle e\right\vert .
\]
Notice that the receiver can perform a measurement $\left\{  \left\vert
0\right\rangle \left\langle 0\right\vert +\left\vert 1\right\rangle
\left\langle 1\right\vert ,\left\vert e\right\rangle \left\langle e\right\vert
\right\}  $ and can learn whether the channel erased the state. The receiver
can do this without disturbing the state in any way. An isometric extension
$U_{\mathcal{N}_{\epsilon}}^{A^{\prime}\rightarrow BE}$\ of it acts as follows
on a purification $\left\vert \psi\right\rangle ^{AA^{\prime}}$\ of the state
$\rho^{A^{\prime}}$:%
\[
U_{\mathcal{N}_{\epsilon}}^{A^{\prime}\rightarrow BE}\left\vert \psi
\right\rangle ^{AA^{\prime}}=\sqrt{1-\epsilon}\left\vert \psi\right\rangle
^{AB}\left\vert e\right\rangle ^{E}+\sqrt{\epsilon}\left\vert \psi
\right\rangle ^{AE}\left\vert e\right\rangle ^{B}.
\]
In the above representation, we see that the erasure channel has the
interpretation that it hands the input to Bob with probability $1-\epsilon$
while giving an erasure flag $\left\vert e\right\rangle $ to Eve, and it hands
the input to Eve with probability $\epsilon$ while giving the erasure flag to Bob.

\begin{theorem}
\label{thm:erasure-theorem}The dynamic capacity region $\mathcal{C}%
_{\mathrm{{CQE}}}(\mathcal{N}_{\epsilon})$ of a quantum erasure channel
$\mathcal{N}_{\epsilon}$ is the set of all $C$, $Q$, and $E$ such that%
\begin{align*}
C+2Q  &  \leq\left(  1-\epsilon\right)  \left(  1+H_{2}\left(  p\right)
\right)  ,\\
Q+E  &  \leq\left(  1-2\epsilon\right)  H_{2}\left(  p\right)  ,\\
C+Q+E  &  \leq1-\epsilon-\epsilon H_{2}\left(  p\right)  ,
\end{align*}
where $p\in\left[  0,1/2\right]  $.
\end{theorem}

Figure~\ref{fig:full-dynamic-erasure}\ plots the dynamic capacity region of a
quantum erasure channel with erasure parameter $\epsilon=1/4$. It turns out
that time-sharing is the optimal strategy here, and there is not an
interesting trade-off curve for the quantum erasure channel.%
\begin{figure}
[ptb]
\begin{center}
\includegraphics[
natheight=4.567100in,
natwidth=7.486700in,
height=2.4768in,
width=4.0413in
]%
{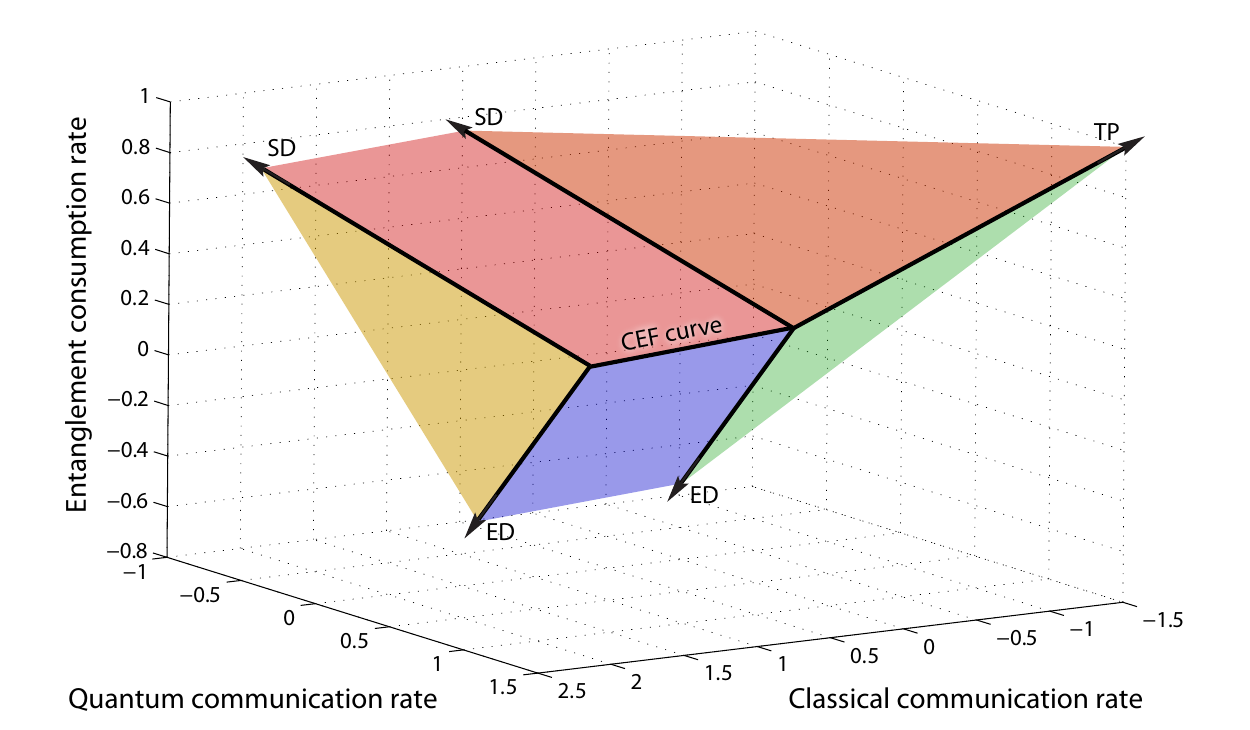}%
\caption{(Color online) A plot of the dynamic capacity region for a qubit
erasure channel with erasure parameter $\epsilon=1/4$. The plot shows that the
classically-enhanced father (CEF) trade-off curve lies along the boundary of
the dynamic capacity region and it is not actually a curve but rather a line
because time-sharing is optimal. The rest of the region is simply the
combination of the CEF points with the unit protocols teleportation (TP),
super-dense coding (SD), and entanglement distribution (ED).}%
\label{fig:full-dynamic-erasure}%
\end{center}
\end{figure}

We in fact proved Theorem~\ref{thm:erasure-theorem}\ in Ref.~\cite{HW09T3}\ by
employing a \textit{reductio ad absurdum} argument reminiscent of the earliest
arguments for proving capacities of quantum erasure
channels~\cite{PhysRevLett.78.3217}. This approach gives the correct answer,
but suffers from two shortcomings:

\begin{enumerate}
\item We do not learn much about how to exploit the structure of the quantum
erasure channel with the \textit{reductio ad absurdum}\ approach. As an
example, Smith and Yard exploited the simple structure of the quantum erasure
channel and discovered far reaching consequences~\cite{SY08}. In particular,
they discovered that the quantum capacity (and for that matter, any future
proposed quantum capacity formula) can never be generally additive, by combining the
erasure channel with another one.

\item The \textit{reductio ad absurdum}\ argument rests on the assumption that
several known capacity formulas are continuous as a function of channels.
Leung and Smith later showed that the known formulas are indeed
continuous~\cite{LS08}, redeeming the original argument in
Ref.~\cite{PhysRevLett.78.3217}.
\end{enumerate}

Here, we prove Theorem~\ref{thm:erasure-theorem}\ above by carefully studying
the structure of the quantum erasure channel and its additivity properties for
the full dynamic capacity region. We prove the theorem in a few steps. First,
we prove that the classical capacity of the quantum erasure channel admits a
single-letter formula.\footnote{The proof already appears in
Ref.~\cite{PhysRevLett.78.3217}, but it again suffers from the aforementioned
shortcomings.} We then simplify the quantum dynamic capacity formula in
(\ref{eq:objective}) for the case of a quantum erasure channel and find that
it is only necessary to consider certain values of the parameters $\mu$ and
$\lambda$ when we optimize. The proof of Lemma~\ref{lem:erasure-base-case}
exploits these conditions to show that the quantum dynamic capacity formula is
additive for the case of two quantum erasure channels. It then follows by a
trivial induction step (the same as in Lemma~\ref{thm:CEQ-single-letter}) that
the full dynamic region single-letterizes and is of the form in
Theorem~\ref{thm:erasure-theorem}.

\begin{lemma}
[Bennett \textit{et al}.~\cite{PhysRevLett.78.3217}]%
\label{lem:classical-cap-erasure}The Holevo information of a quantum
erasure channel is equal to $1-\epsilon$:%
\[
\chi\left(  \mathcal{N}_{\epsilon}\right)  \equiv\max_{\sigma^{XB}}I\left(
X;B\right)  =1-\epsilon
\]

\end{lemma}

\begin{proof}
We begin with an input ensemble of the following form:%
\[
\rho^{XA^{\prime}}\equiv\sum_{x}p_{X}\left(  x\right)  \left\vert
x\right\rangle \left\langle x\right\vert ^{X} \otimes\phi_{x}^{A^{\prime}},
\]
where the states $\phi_{x}^{A^{\prime}}$ are pure (it suffices to consider
pure state ensembles for classical capacity). Feeding the $A^{\prime}$ system
into the quantum erasure channel leads to a classical-quantum state
$\sigma^{XB}=\sum_{x}p_{X}\left(  x\right)  \left\vert x\right\rangle
\left\langle x\right\vert ^{X}\otimes\mathcal{N}_{\epsilon}^{A^{\prime}\to
B}(\phi_{x}^{A^{\prime}})$. Bob can then measure his system $B$ to learn
whether the channel erases the qubit. Let $X_{E}$ denote a classical register
where Bob places the result of the measurement so that we then have the state
$\sigma^{XBX_{E}}$. It holds that any entropy evaluated on system $B$ is equal
to the joint entropy of $B$ and $X_{E}$ because this measurement does not
disturb the state in any way. Consider the following chain of inequalities:
\begin{align*}
I\left(  X;B\right)  _{\sigma}  &  =H\left(  B\right)  _{\sigma}-H\left(
B|X\right)  _{\sigma}\\
&  =H\left(  BX_{E}\right)  _{\sigma}-H\left(  BX_{E}|X\right)  _{\sigma}\\
&  =H\left(  B|X_{E}\right)  _{\sigma}-H\left(  B|X_{E}X\right)  _{\sigma}\\
&  =\left(  1-\epsilon\right)  H\left(  A^{\prime}\right)  _{\rho}-\left(
1-\epsilon\right)  H\left(  A^{\prime}|X\right)  _{\rho}\\
&  =\left(  1-\epsilon\right)  I\left(  X;A^{\prime}\right)  _{\rho}\\
&  \leq\left(  1-\epsilon\right)  .
\end{align*}
The first equality follows by expanding the mutual information and the second
equality follows from the above fact regarding the joint entropy of $B$ and
$X_{E}$. The third equality follows by expanding and canceling terms. The
fourth equality follows by conditioning on the classical erasure flag register
$X_{E}$ and realizing that the entropy of Bob's system $B$\ is the entropy of
the input state with probability $1-\epsilon$ and otherwise is the entropy of
the erasure state $\left\vert e\right\rangle $ (this latter entropy vanishes
because this state is pure). Thus, the above sequence of steps reduces an
optimization problem on the output of the channel to a simple optimization
over the input ensemble. The Holevo information is then $1-\epsilon$ because
the quantity $I\left(  X;A^{\prime}\right)  _{\rho}$ can never be larger than
unity for the case of a qubit erasure channel (it reaches unity for an
ensemble of two orthogonal pure states chosen with equal probability).
\end{proof}

\begin{lemma}
\label{thm:additive-classical-erasure}The following additivity lemma holds for
a quantum erasure channel $\mathcal{N}_{\epsilon}$:%
\[
\chi\left(  \mathcal{N}_{\epsilon}\otimes\mathcal{N}_{\epsilon}\right)
=\chi\left(  \mathcal{N}_{\epsilon}\right)  +\chi\left(  \mathcal{N}%
_{\epsilon}\right)  =2\left(  1-\epsilon\right)  .
\]

\end{lemma}

\begin{proof}
It suffices to prove the inequality $\chi\left(  \mathcal{N}_{\epsilon}%
\otimes\mathcal{N}_{\epsilon}\right)  \leq\chi\left(  \mathcal{N}_{\epsilon
}\right)  +\chi\left(  \mathcal{N}_{\epsilon}\right)  $ because the other
inequality holds trivially. We define the following ensemble of states for the
tensor product channel $\mathcal{N}_{\epsilon}\otimes\mathcal{N}_{\epsilon}$:%
\begin{equation}
\rho^{XA_{1}^{\prime}A_{2}^{\prime}}\equiv\sum_{x}p_{X}\left(  x\right)
\left\vert x\right\rangle \left\langle x\right\vert ^{X}\otimes\phi_{x}%
^{A_{1}^{\prime}A_{2}^{\prime}}. \label{eq:cq-state-erasure}%
\end{equation}
We can also define the following augmented ensemble based on the above one:%
\begin{equation}
\sigma^{XIJA_{1}^{\prime}A_{2}^{\prime}}\equiv\frac{1}{16}\sum_{x,i,j}%
p_{X}\left(  x\right)  \left\vert x\right\rangle \left\langle x\right\vert
^{X}\otimes\left\vert i\right\rangle \left\langle i\right\vert ^{I}%
\otimes\left\vert j\right\rangle \left\langle j\right\vert ^{J}\otimes
(\sigma_{i}^{A_{1}^{\prime}}\otimes\sigma_{j}^{A_{2}^{\prime}})\phi_{x}%
^{A_{1}^{\prime}A_{2}^{\prime}}(\sigma_{i}^{A_{1}^{\prime}}\otimes\sigma
_{j}^{A_{2}^{\prime}}), \label{eq:cq-state-erasure-mixed}%
\end{equation}
where $\sigma_{0}\equiv I$, $\sigma_{1}\equiv\sigma_{X}$, $\sigma_{2}%
\equiv\sigma_{Y}$, and $\sigma_{3}\equiv\sigma_{Z}$. In particular, note that
we obtain the maximally mixed state when tracing over classical registers $I$
and $J$. Let $\omega^{XB_{1}B_{2}}$ and $\theta^{XIJB_{1}B_{2}}$ denote the
states obtained by sending systems $A_{1}^{\prime}$ and $A_{2}^{\prime}$ of
the respective states $\rho^{XA_{1}^{\prime}A_{2}^{\prime}}$ and
$\sigma^{XIJA_{1}^{\prime}A_{2}^{\prime}}$ through two uses of the quantum
erasure channel. Let $X_{E,1}$ and $X_{E,2}$ denote the classical variables
Bob obtains by determining whether the channel erased his states (they also
denote the registers where he places the results). Consider the following
chain of inequalities that holds for any state $\omega^{XB_{1}B_{2}}$:%
\begin{align*}
I\left(  X;B_{1}B_{2}\right)  _{\omega}  &  =H\left(  B_{1}B_{2}\right)
_{\omega}-H\left(  B_{1}B_{2}|X\right)  _{\omega}\\
&  =H\left(  B_{1}B_{2}|X_{E,1}X_{E\,,2}\right)  _{\omega}-H\left(  B_{1}%
B_{2}|X_{E,1}X_{E\,,2}X\right)  _{\omega}\\
&  =\left(  1-\epsilon\right)  ^{2}H\left(  A_{1}^{\prime}A_{2}^{\prime
}\right)  _{\rho}+\left(  1-\epsilon\right)  \epsilon\left(  H\left(
A_{1}^{\prime}\right)  _{\rho}+H\left(  A_{2}^{\prime}\right)  _{\rho}\right)
\\
&  \ \ \ \ \ \ -\left(  1-\epsilon\right)  ^{2}H\left(  A_{1}^{\prime}%
A_{2}^{\prime}|X\right)  _{\rho}-\left(  1-\epsilon\right)  \epsilon\left(
H\left(  A_{1}^{\prime}|X\right)  _{\rho}+H\left(  A_{2}^{\prime}|X\right)
_{\rho}\right) \\
&  \leq2\left(  1-\epsilon\right)  -\left(  1-\epsilon\right)  ^{2}H\left(
A_{1}^{\prime}A_{2}^{\prime}|XIJ\right)  _{\sigma}-\left(  1-\epsilon\right)
\epsilon\left(  H\left(  A_{1}^{\prime}|XIJ\right)  _{\sigma}+H\left(
A_{2}^{\prime}|XIJ\right)  _{\sigma}\right) \\
&  \leq2\left(  1-\epsilon\right)  .
\end{align*}
The first equality holds by expanding the mutual information. The next
equality holds because an \textquotedblleft erasure
measurement\textquotedblright\ does not change entropy. The third equality
follows by exploiting the properties of the erasure channels. The first
inequality holds because the unconditional entropies of the $A^{\prime}$
systems on the state $\rho$ are always less than those for the state $\sigma$.
The final inequality follows because the entropies in the previous line are
non-negative. The statement of the theorem then follows.
\end{proof}

\begin{lemma}
The quantum dynamic capacity formula in (\ref{eq:objective})\ simplifies as
follows for a quantum erasure channel~$\mathcal{N}_{\epsilon}$:%
\begin{equation}
D_{\lambda,\mu}\left(  \mathcal{N}_{\epsilon}\right)  \equiv\max_{p\in\left[
0,1/2\right]  }\left(  1-\epsilon\right)  \left(  1+H_{2}\left(  p\right)
\right)  +\lambda\left(  1-2\epsilon\right)  H_{2}\left(  p\right)
+\mu\left(  \left(  1-\epsilon\right)  -\epsilon H_{2}\left(  p\right)
\right)  . \label{eq:objective-erasure}%
\end{equation}
Thus, the \textquotedblleft one-letter\textquotedblright\ dynamic capacity
region of a quantum erasure channel is as Theorem~\ref{thm:erasure-theorem} states.
\end{lemma}

\begin{proof}
We exploit the following classical-quantum states:%
\begin{align*}
\rho^{XAA^{\prime}}  &  \equiv\sum_{x}p_{X}\left(  x\right)  \left\vert
x\right\rangle \left\langle x\right\vert ^{X}\otimes\phi_{x}^{AA^{\prime}},\\
\sigma^{XIAA^{\prime}}  &  \equiv\sum_{x,i}\frac{1}{4}p_{X}\left(  x\right)
\left\vert x\right\rangle \left\langle x\right\vert ^{X}\otimes\left\vert
i\right\rangle \left\langle i\right\vert ^{I}\otimes(\sigma_{i}^{A^{\prime}%
})\phi_{x}^{AA^{\prime}}(\sigma_{i}^{A^{\prime}}),
\end{align*}
and let $\rho^{XABE}$ and $\sigma^{XIABE}$ be the states obtained by
transmitting the $A^{\prime}$ system through the isometric extension of the
erasure channel. Let $\sigma_{x}^{A^{\prime}}\equiv$Tr$_{A}\left\{  \phi
_{x}^{AA^{\prime}}\right\}  $. Furthermore, let the eigenvalues of the state
$\sigma_{x}^{A^{\prime}}$ with highest entropy on system $A^{\prime}$ be $p$
and $1-p$. Consider that the following chain of inequalities holds for any
state $\rho^{XABE}$:%
\begin{align*}
&  I\left(  AX;B\right)  _{\rho}+\lambda I\left(  A\rangle BX\right)  _{\rho
}+\mu\left(  I\left(  X;B\right)  _{\rho}+I\left(  A\rangle BX\right)  _{\rho
}\right) \\
&  =H\left(  A|X\right)  _{\rho}+\left(  \mu+1\right)  H\left(  B\right)
_{\rho}+\lambda H\left(  B|X\right)  _{\rho}-\left(  \lambda+\mu+1\right)
H\left(  E|X\right)  _{\rho}\\
&  =H\left(  A^{\prime}|X\right)  _{\rho}+\left(  \mu+1\right)  H\left(
B|X_{E}\right)  _{\rho}+\lambda H\left(  B|X_{E}X\right)  _{\rho}-\left(
\lambda+\mu+1\right)  H\left(  E|X_{E}X\right)  _{\rho}\\
&  =H\left(  A^{\prime}|X\right)  _{\rho}+\left(  \mu+1\right)  \left(
1-\epsilon\right)  H\left(  A^{\prime}\right)  _{\rho}+\lambda\left(
1-\epsilon\right)  H\left(  A^{\prime}|X\right)  _{\rho}-\left(  \lambda
+\mu+1\right)  \epsilon H\left(  A^{\prime}|X\right)  _{\rho}\\
&  \leq\left(  \mu+1\right)  \left(  1-\epsilon\right)  +\left(
1+\lambda\left(  1-\epsilon\right)  -\left(  \lambda+\mu+1\right)
\epsilon\right)  H\left(  A^{\prime}|XI\right)  _{\sigma}\\
&  =\left(  \mu+1\right)  \left(  1-\epsilon\right)  +\left(  1-\epsilon
+\lambda\left(  1-2\epsilon\right)  -\mu\epsilon\right)  \sum_{x}p_{X}\left(
x\right)  H\left(  A^{\prime}\right)  _{\sigma_{x}}\\
&  \leq\left(  \mu+1\right)  \left(  1-\epsilon\right)  +\left(
1-\epsilon+\lambda\left(  1-2\epsilon\right)  -\mu\epsilon\right)  H\left(
A^{\prime}\right)  _{\sigma_{x}^{\ast}}\\
&  =\left(  1-\epsilon\right)  \left(  1+H_{2}\left(  p\right)  \right)
+\mu\left(  1-\epsilon-\epsilon H_{2}\left(  p\right)  \right)  +\lambda
\left(  1-2\epsilon\right)  H_{2}\left(  p\right)  .
\end{align*}
The first equality follows by standard entropic manipulations. The second
equality follows by incorporating the classical erasure flag variable. The
third equality follows by exploiting the properties of the quantum erasure
channel. The first inequality follows because the unconditional entropy of the
state $\rho$ is always less than that of the state $\sigma$. The next equality
follows by expanding the conditional entropy. The second inequality follows
because an average is always less than a maximum. The final equality follows
by plugging in the eigenvalues of $\sigma_{x}^{\ast}$. The form of the quantum
dynamic capacity formula then follows because this chain of inequalities holds
for any input ensemble.
\end{proof}

\begin{lemma}
\label{lem:suff-condition}It suffices to consider the set of $\lambda,\mu
\geq0$ for which%
\[
\left(  1-\epsilon\right)  +\lambda\left(  1-2\epsilon\right)  \geq\mu
\epsilon.
\]
Otherwise, we are just maximizing the classical capacity, which we know from
Lemma~\ref{lem:classical-cap-erasure}\ is equal to $1-\epsilon$.
\end{lemma}

\begin{proof}
Consider rewriting the expression in (\ref{eq:objective-erasure}) as follows:%
\[
\max_{p\in\left[  0,1/2\right]  }\left(  1-\epsilon\right)  +\mu\left(
1-\epsilon\right)  +\left[  \left(  1-\epsilon\right)  +\lambda\left(
1-2\epsilon\right)  -\mu\epsilon\right]  H_{2}\left(  p\right)  .
\]
Suppose that the expression in square brackets is negative, i.e.,%
\[
\left(  1-\epsilon\right)  +\lambda\left(  1-2\epsilon\right)  <\mu\epsilon.
\]
Then the maximization over $p$ simply chooses $p=0$ so that $H_{2}\left(
p\right)  $ vanishes and the negative term disappears. The resulting
expression for the quantum dynamic capacity formula is%
\[
\left(  1-\epsilon\right)  +\mu\left(  1-\epsilon\right)  ,
\]
which corresponds to the following region%
\begin{align*}
C+2Q  &  \leq1-\epsilon,\\
Q+E  &  \leq0,\\
C+Q+E  &  \leq1-\epsilon.
\end{align*}
The above region is equivalent to a translation of the unit resource capacity
region to the classical capacity rate triple $\left(  1-\epsilon,0,0\right)
$. Thus, it suffices to restrict the parameters $\lambda$ and $\mu$ as above
for the quantum erasure channel.
\end{proof}

\begin{lemma}
\label{lem:erasure-base-case}The following additivity relation holds for two
quantum erasure channels $\mathcal{N}_{\epsilon}$ with the same erasure
parameter $\epsilon$:%
\[
D_{\lambda,\mu}(\mathcal{N}_{\epsilon}\otimes\mathcal{N}_{\epsilon
})=D_{\lambda,\mu}(\mathcal{N}_{\epsilon})+D_{\lambda,\mu}(\mathcal{N}%
_{\epsilon}).
\]

\end{lemma}

\begin{proof}
We prove the non-trivial inequality $D_{\lambda,\mu}(\mathcal{N}_{\epsilon
}\otimes\mathcal{N}_{\epsilon})\leq D_{\lambda,\mu}(\mathcal{N}_{\epsilon
})+D_{\lambda,\mu}(\mathcal{N}_{\epsilon})$. We define the following states:%
\begin{align*}
\rho^{XAA_{1}^{\prime}A_{2}^{\prime}}  &  \equiv\sum_{x}p_{X}\left(  x\right)
\left\vert x\right\rangle \left\langle x\right\vert ^{X}\otimes\phi
_{x}^{AA_{1}^{\prime}A_{2}^{\prime}},\\
\omega^{XAB_{1}E_{1}B_{2}E_{2}}  &  \equiv U_{\mathcal{N}_{\epsilon}}%
^{A_{1}^{\prime}\rightarrow B_{1}E_{1}}\otimes U_{\mathcal{N}_{\epsilon}%
}^{A_{2}^{\prime}\rightarrow B_{2}E_{2}}(\rho^{XAA_{1}^{\prime}A_{2}^{\prime}%
}),
\end{align*}
and we suppose that $\rho^{XAA_{1}A_{2}}$ is the state that maximizes
$D_{\lambda,\mu}(\mathcal{N}_{\epsilon}\otimes\mathcal{N}_{\epsilon})$.
Consider the following equality:%
\begin{align*}
&  I\left(  AX;B_{1}B_{2}\right)  _{\omega}+\lambda I\left(  A\rangle
B_{1}B_{2}X\right)  _{\omega}+\mu\left(  I\left(  X;B_{1}B_{2}\right)
_{\omega}+I\left(  A\rangle B_{1}B_{2}X\right)  _{\omega}\right) \\
&  =H\left(  A|X\right)  _{\omega}+H\left(  B_{1}B_{2}\right)  _{\omega
}-H\left(  E_{1}E_{2}|X\right)  _{\omega}+\lambda\left(  H\left(  B_{1}%
B_{2}|X\right)  _{\omega}-H\left(  E_{1}E_{2}|X\right)  _{\omega}\right) \\
&  \ \ \ \ \ \ \ \ +\mu\left(  H\left(  B_{1}B_{2}\right)  _{\omega}-H\left(
E_{1}E_{2}|X\right)  _{\omega}\right)  .
\end{align*}
It follows simply by rewriting entropies. We continue below:%
\begin{align*}
&  =H\left(  A_{1}^{\prime}A_{2}^{\prime}|X\right)  _{\rho}+\left(
1-\epsilon\right)  ^{2}H\left(  A_{1}^{\prime}A_{2}^{\prime}\right)  _{\rho
}+\epsilon\left(  1-\epsilon\right)  \left(  H\left(  A_{1}^{\prime}\right)
_{\rho}+H\left(  A_{2}^{\prime}\right)  _{\rho}\right) \\
&  \ \ \ \ \ \ -\left[  \epsilon^{2}H\left(  A_{1}^{\prime}A_{2}^{\prime
}|X\right)  _{\rho}+\epsilon\left(  1-\epsilon\right)  \left(  H\left(
A_{1}^{\prime}|X\right)  _{\rho}+H\left(  A_{2}^{\prime}|X\right)  _{\rho
}\right)  \right] \\
&  \ \ \ \ \ \ +\lambda\left[  \left(  1-\epsilon\right)  ^{2}H\left(
A_{1}^{\prime}A_{2}^{\prime}|X\right)  _{\rho}+\epsilon\left(  1-\epsilon
\right)  \left(  H\left(  A_{1}^{\prime}|X\right)  _{\rho}+H\left(
A_{2}^{\prime}|X\right)  _{\rho}\right)  \right] \\
&  \ \ \ \ \ \ -\lambda\left[  \epsilon^{2}H\left(  A_{1}^{\prime}%
A_{2}^{\prime}|X\right)  _{\rho}+\epsilon\left(  1-\epsilon\right)  \left(
H\left(  A_{1}^{\prime}|X\right)  _{\rho}+H\left(  A_{2}^{\prime}|X\right)
_{\rho}\right)  \right] \\
&  \ \ \ \ \ \ +\mu\left[  \left(  1-\epsilon\right)  ^{2}H\left(
A_{1}^{\prime}A_{2}^{\prime}\right)  _{\rho}+\epsilon\left(  1-\epsilon
\right)  \left(  H\left(  A_{1}^{\prime}\right)  _{\rho}+H\left(
A_{2}^{\prime}\right)  _{\rho}\right)  \right] \\
&  \ \ \ \ \ \ -\mu\left[  \epsilon^{2}H\left(  A_{1}^{\prime}A_{2}^{\prime
}|X\right)  _{\rho}+\epsilon\left(  1-\epsilon\right)  \left(  H\left(
A_{1}^{\prime}|X\right)  _{\rho}+H\left(  A_{2}^{\prime}|X\right)  _{\rho
}\right)  \right]  .
\end{align*}
The above equality follows by exploiting the properties of the quantum erasure
channel. Continuing, the above quantity is less than the following one:%
\begin{align*}
&  \leq2\left(  1-\epsilon\right)  +\left(  1-\epsilon^{2}\right)  H\left(
A_{1}^{\prime}A_{2}^{\prime}|XIJ\right)  _{\sigma}+\epsilon\left(
1-\epsilon\right)  \left(  H\left(  A_{1}^{\prime}|XIJ\right)  _{\sigma
}+H\left(  A_{2}^{\prime}|XIJ\right)  _{\sigma}\right) \\
&  \ \ \ \ \ \ +\lambda\left(  1-2\epsilon\right)  H\left(  A_{1}^{\prime
}A_{2}^{\prime}|XIJ\right)  _{\sigma}\\
&  \ \ \ \ \ \ +\mu\left[  2\left(  1-\epsilon\right)  -\epsilon^{2}H\left(
A_{1}^{\prime}A_{2}^{\prime}|XIJ\right)  _{\sigma}-\epsilon\left(
1-\epsilon\right)  \left(  H\left(  A_{1}^{\prime}|XIJ\right)  _{\sigma
}-H\left(  A_{2}^{\prime}|XIJ\right)  _{\sigma}\right)  \right] \\
&  =2\left(  1-\epsilon\right)  +\left(  1-\epsilon\right)  \left(  H\left(
A_{1}^{\prime}|XIJ\right)  _{\sigma}+H\left(  A_{2}^{\prime}|XIJ\right)
_{\sigma}\right)  +\lambda\left(  1-2\epsilon\right)  \left(  H\left(
A_{1}^{\prime}|XIJ\right)  _{\sigma}+H\left(  A_{2}^{\prime}|XIJ\right)
_{\sigma}\right) \\
&  \ \ \ \ \ \ +\mu\left[  2\left(  1-\epsilon\right)  -\epsilon\left(
H\left(  A_{1}^{\prime}|XIJ\right)  _{\sigma}-H\left(  A_{2}^{\prime
}|XIJ\right)  _{\sigma}\right)  \right] \\
&  \ \ \ \ \ \ -\left[  \left(  1-\epsilon^{2}+\lambda\left(  1-2\epsilon
\right)  -\mu\epsilon^{2}\right)  I\left(  A_{1}^{\prime};A_{2}^{\prime
}|XIJ\right)  _{\sigma}\right] \\
&  \leq D_{\lambda,\mu}\left(  \mathcal{N}_{\epsilon}\right)  +D_{\lambda,\mu
}\left(  \mathcal{N}_{\epsilon}\right)  -\left[  \left(  1-\epsilon
^{2}+\lambda\left(  1-2\epsilon\right)  -\mu\epsilon^{2}\right)  I\left(
A_{1}^{\prime};A_{2}^{\prime}|XIJ\right)  _{\rho}\right] \\
&  \leq D_{\lambda,\mu}\left(  \mathcal{N}_{\epsilon}\right)  +D_{\lambda,\mu
}\left(  \mathcal{N}_{\epsilon}\right)  .
\end{align*}
The first inequality follows from similar proofs we have seen for a state
$\sigma$\ of the form in (\ref{eq:cq-state-erasure-mixed}). The first equality
follows by rearranging terms. The second inequality follows from the form of
$D_{\lambda,\mu}$ in (\ref{eq:objective-erasure}). The final inequality
follows because Lemma~\ref{lem:suff-condition} states that it is sufficient to
consider $\left(  1-\epsilon\right)  +\lambda\left(  1-2\epsilon\right)
\geq\mu\epsilon$. Note that this condition implies that%
\[
1-\epsilon^{2}+\lambda\left(  1-2\epsilon\right)  \geq\mu\epsilon^{2},
\]
and hence that the quantity in square brackets in the line above the last one
is positive.
\end{proof}

\section{Conclusion}

We found a purely information theoretic approach to proving the converse part
of the dynamic capacity region. This technique should be simpler to understand
for those unfamiliar with the quantum Shannon theory literature. We also
phrased the optimization task for the full dynamic capacity region in terms of
the quantum dynamic capacity formula in (\ref{eq:objective}) and proved its
additivity (and hence single-letterization of the dynamic capacity region) for
the quantum Hadamard channels and quantum erasure channels. We note some open
problems below.

There might be room for improvement in our formulas that characterize the
dynamic capacity region when the channel is not of the Hadamard class or a
quantum erasure channel. Though, our characterization has the simple
interpretation as the regularization of what one can achieve with the
classically-enhanced father protocol~\cite{HW08GFP}\ combined with
teleportation, super-dense coding, and entanglement distribution. It is
difficult to imagine a simpler characterization than this one, despite its
multi-letter nature for the general case.

We would like to find other channels besides the Hadamard or erasure channels
for which the region single-letterizes. We conjecture that additivity of the
quantum dynamic capacity formula in (\ref{eq:objective}) holds for channels
that have hybrid Hadamard-erasure behavior such as the phase erasure channel
in Ref.~\cite{PhysRevLett.78.3217}. It would also be interesting to find
channels that are not hybrid Hadamard-erasure for which additivity of
(\ref{eq:objective}) holds.

There is also one interesting speculation to muse over that Professor David
Avis suggested to us. Do each of the inequalities in
Theorem~\ref{thm:main-theorem}\ correspond to some fundamental physical law?
This might shed further connections between information theory and physics
that have not been elucidated yet.

\section*{Acknowledgements}

We acknowledge the anonymous referee of Ref.~\cite{HW09T3}, who encouraged us
to find simpler proofs of the triple trade-off capacity regions, and the
anonymous referee of Ref.~\cite{HW08GFP}\ who questioned what the optimization
task was for entanglement-assisted communication of classical and quantum
information. MMW\ also acknowledges all of the useful discussions with Kamil
Br\'{a}dler, Patrick Hayden, and Dave Touchette during the development of
Ref.~\cite{BHTW10} and useful discussions with David Avis and Patrick Hayden
concerning the optimization task. We acknowledge Patrick Hayden for suggesting
the catalytic approach for proving the converse theorem. MMW acknowledges
support from the MDEIE (Qu\'{e}bec) PSR-SIIRI international collaboration grant.


\end{document}